\documentclass{article}
\usepackage[left=1in, right=1in, bottom = 1in, top = 1in]{geometry}
\usepackage{graphicx} 
\usepackage{bbm}
\usepackage{amsmath}
\usepackage[english]{babel}
\usepackage[utf8]{inputenc}
\usepackage[authoryear]{natbib}
\usepackage{setspace}

\usepackage{amsmath}
\usepackage{amsfonts}
\usepackage{amsthm}
\usepackage{amssymb}
\usepackage{mathrsfs}
\usepackage{mathtools}
\usepackage{amstext}
\usepackage{comment}
\usepackage{epsfig}
\usepackage{latexsym}
\usepackage{epstopdf}
\usepackage{bm}
\usepackage{algorithm}
\usepackage{algpseudocode}

\usepackage{tikz}
\usetikzlibrary{positioning,chains,fit,shapes,calc}
\definecolor{myblue}{RGB}{80,80,160}
\definecolor{mygreen}{RGB}{80,160,80}

\usepackage{url}
\usepackage[colorlinks = true, linkcolor = blue]{hyperref}
\hypersetup{
        urlcolor = blue,
	colorlinks = true,
	linkcolor = blue,
	citecolor = blue,
}
\usepackage{paralist}
\usepackage{enumitem}
\setlist[enumerate]{topsep=0pt,noitemsep}

\usepackage{multirow}
\usepackage{multicol}
\usepackage{babel}
\usepackage{array}
\usepackage{rotating}
\usepackage{longtable}
\usepackage{float}
\usepackage{booktabs}
\usepackage{lscape}
\usepackage{caption}
\usepackage{subcaption}
\usepackage{afterpage}
\usepackage{authblk}

\usepackage{array,booktabs}
\newcolumntype{C}[1]{>{\centering\arraybackslash}p{#1}}
\newcommand{\colW}{1.15cm}
\newcommand{\colWw}{1cm}

\newtheorem*{design}{Design Mechanism}
\newtheorem{theorem}{Theorem}[section]
\newtheorem{remark}{Remark}[section]
\newtheorem{assumption}{Assumption}
\newtheorem{lemma}{Lemma}[section]

\newtheorem{corollary}{Corollary}[section]

\allowdisplaybreaks

\def\T{{\textsc{t}}}

\newcommand{\pr}{\textup{pr}}
\newcommand{\cov}{\textup{cov}}
\newcommand{\var}{\textup{var}}
\newcommand{\avar}{\textup{avar}}
\newcommand{\acov}{\textup{acov}}

\newcommand{\bOneA}{\boldsymbol{1}_{n}}

\newcommand{\bY}{\boldsymbol{Y}}

\newcommand{\ba}{\boldsymbol{a}}
\newcommand{\bA}{\boldsymbol{A}}
\newcommand{\bB}{\boldsymbol{B}}

\newcommand{\bOmega}{\boldsymbol{\Omega}}
\newcommand{\bLambda}{\boldsymbol{\Lambda}}

\newcommand{\cN}{\mathcal{N}}

\newcommand{\cB}{\mathcal{B}}

\newcommand{\bGamma}{\boldsymbol{\Gamma}}

\newcommand{\cG}{\mathcal{G}}

\newcommand{\sumi}{\sum_{i=1}^n}
\newcommand{\sumj}{\sum_{j=1}^n}

\newcommand{\sumij}{\sum_{i=1}^{n} \sum_{j=1}^{n}}
\newcommand{\sumjneqi}{\sum_{j\neq i}}

\newcommand{\sumrs}{\sum_{r=1}^{n} \sum_{s=1}^{n}}
\newcommand{\indep}{\perp\!\!\!\perp}

\newcommand{\de}{\textsc{de}}
\newcommand{\bie}{\textsc{bie}}
\newcommand{\wie}{\textsc{wie}}
\newcommand{\horthom}{\textup{ht}}
\newcommand{\haj}{\textup{haj}}
\newcommand{\diag}{\textup{diag}}

\newcommand{\indicator}{\mathbbm{1}}
\newcommand{\bbI}{\mathbbm{I}}
\newcommand{\ashprime}{a^{\prime},s^{\prime},h^{\prime}}

\newcommand{\pscore}[1]{\pi_{#1}}

\newcommand{\nonindep}{{~\not\!\perp\!\!\!\perp~}}

\makeatletter
\newcommand{\ostar}{\mathbin{\mathpalette\make@circled\star}}
\newcommand{\make@circled}[2]{%
  \ooalign{$\m@th#1\smallbigcirc{#1}$\cr\hidewidth$\m@th#1#2$\hidewidth\cr}%
}
\newcommand{\smallbigcirc}[1]{%
  \vcenter{\hbox{\scalebox{0.77778}{$\m@th#1\bigcirc$}}}%
}
\makeatother

\allowdisplaybreaks

\title{Estimating within-cluster and between-cluster spillover effects in randomized saturation designs}

\author[1]{Sizhu Lu\thanks{The first two authors contributed equally to this work. This research was partially supported by the U.S. National Science Foundation (1745640, 1945136, 2514234). We also thank the participants of Berkeley Causal Lab for their valuable suggestions. }}
\author[2]{Lei Shi$^{*}$}
\author[1]{Peng Ding}

\affil[1]{Department of Statistics, University of California, Berkeley}
\affil[2]{Division of Biostatistics, University of California, Berkeley}

\date{\vspace{-2em}}
\begin{document}

\maketitle
\begin{abstract}
    Randomized saturation designs are two-stage experiments: they first randomly assign treatment probabilities over the clusters and then randomly assign the treatment to the units within the clusters. The existing literature on randomized saturation designs focuses on estimating within-cluster spillover effects by assuming away between-cluster spillover effects. However, the units may interact across clusters in many practical randomized saturation designs. A leading example is that some units are geographically close to each other, so spillover effects arise across clusters. Based on the potential outcomes framework, we formulate the causal inference problem of estimating within-cluster and between-cluster spillover effects in randomized saturation designs. We clarify the causal estimands and establish the statistical theory for estimation and inference. We also apply our method to analyze a recent randomized saturation design of cash transfer on household expenditure in Kenya.
\end{abstract}
{\small\hspace{2em} \textbf{Key Words: } direct effect;  exposure mapping; indirect effect; interference; two-stage experiment}

\section{Introduction}
Randomized saturation designs have become increasingly common across disciplines for studying spillovers and interference, with applications spanning economics (e.g., \citealt{crepon2013labor, baird2018optimal, egger2022general}), public health (e.g., \citealt{melis2005design, benjamin2018spillover}), and political science (e.g., \citealt{sinclair2012detecting}). Most existing studies consider only within-cluster interference and assume there is no interference between clusters \citep{hudgens2008toward, basse2018analyzing, jiang2023statistical}. 

However, in some real-world settings, this assumption may not hold. A leading motivating example is the study by \citet{egger2022general}, which implemented a randomized saturation design to evaluate the economic impacts of a large-scale cash transfer program in rural Kenya between 2014 and 2017. The study area consists of 653 villages nested within 155 sublocations across two counties. Sublocations are administrative units, and villages within the same sublocation often share common markets, social ties, and economic connections. 

The randomization proceeded in two stages following a randomized saturation design: 
\begin{enumerate}
    \item in the first stage, sublocations were randomly assigned to high or low saturation levels. 
    
    \item in the second stage, within each high saturation sublocation, two-thirds of villages were randomly assigned to treatment, while in low saturation sublocations, one-third of villages were randomly assigned to treatment. All eligible households in treated villages received transfers.
\end{enumerate}

In this design, there are two possible types of interference:
\begin{itemize}
    \item within-sublocation interference, where the outcome of a village may be affected by the treatment status of other villages in the same sublocation, and
    \item between-sublocation interference, where the outcome of a village may also be affected by the treatment of villages from different sublocations that are geographically close.
\end{itemize}

To make the ideas precise, we start by introducing some notation. We focus on a finite population of $n$ units, indexed by $i$. In the cash transfer study, units are villages grouped into administrative sublocations. Let $k_i$ denote the sublocation to which village $i$ belongs. Some villages are close to one another, and we use $\cG_{i}$ to represent the set of villages that are geographically close (for example, within a certain distance) to village $i$. 
Let $A_i\in\{0,1\}$ denote the binary treatment of each village, with $A_i=1$ if the village $i$ was assigned a cash transfer, and $A_i=0$ otherwise. Let $\bA=(A_1, \ldots, A_n)$ denote the vector of treatment assignment of all villages. We will focus on the following randomized saturation design, which is adopted in \cite{egger2022general}:

\begin{design}
    The randomized saturation design we analyze includes two stages. In the first stage, the sublocations are randomly assigned to high- and low-saturation clusters with probability $1/2$ each. In the second stage, for villages in sublocations assigned to high saturation clusters, $A_i=1$ with probability $2/3$, and for villages in sublocations assigned to low saturation clusters, $A_i=1$ with probability $1/3$. 
\end{design}

When there is interference, a village's outcome, denoted by $Y_i$, may depend not only on its own treatment but also on the treatments received by other villages. Under the potential outcomes framework \citep{neyman1923application,rubin1974estimating}, the potential outcome for village $i$ is written as $Y_i(\ba)$ for $\ba=(a_1,\ldots,a_n)$, generally depending on the $n$-dimensional treatment assignment vector. To make this dependence more tractable, we use the concept of an exposure mapping \citep{aronow2017estimating}, which summarizes the aspects of the treatment assignment relevant for village $i$'s outcome.

In our setting, an exposure mapping summarizes how the treatments received by other villages, together with village $i$’s own treatment, combine to determine how much exposure village $i$ experiences that affects its outcome. Specifically, we define an exposure mapping 
\[
d_i(\bA) = (A_{i}, S_{i}, H_{i}),
\] 
where $A_i$ is the village $i$'s own treatment, $S_i$ summarizes the treatment status of other villages in the same sublocation as village $i$, capturing within-sublocation exposure, and $H_i$ summarizes the treatment status of nearby villages located in different sublocations, capturing between-sublocation exposure. $d_i(\cdot)$ is also a function of the administrative sublocation and geographical relationship, which are treated as fixed. The randomized saturation design determines $A_i$. We define $S_i$ and $H_i$ below. For each village $i$, we write $S_i$ as
\[
S_{i} = f_i\{A_{j}: k_j=k_i, j\neq i\},
\] 
which is a function of the treatment assignments of other villages within the same sublocation. The mapping $f_i$ is allowed to vary across units because different villages may differ in which within-sublocation neighboring treatments matter and how those treatments are aggregated into their potential outcomes. Similarly, we write $H_i$ as
\[
H_{i} = g_i\{A_j: k_j\neq k_i, ~ j\in\cG_{i}\},
\]
which is a function of the treatment assignments of nearby villages outside $i$'s sublocation. The mapping $g_i$ may also vary across units, allowing the relevant between-sublocation exposure structure to vary across villages. Under a correctly specified exposure mapping $d_i(\cdot)$, we assume that $Y_i(\ba)=Y_i(\ba^{\prime})$ whenever $d_i(\ba)=d_i(\ba^{\prime})$. That is, once the exposure mapping values are fixed, further details of the full assignment vector do not affect village $i$'s outcome. We may therefore write the potential outcome as $Y_i(a,s,h)$ for $(a,s,h)$ in the support of the exposure mapping.

As an example, consider the following exposure mapping:  
\begin{eqnarray}
    S_{i} &=& \indicator\left\{\frac{\sum_{j\neq i}\indicator\{k_j=k_i\} A_{j}}{\sum_{j\neq i}\indicator\{k_j=k_i\}} > \frac{1}{2}\right\}, \label{eqn::Si_example} \\
    H_{i} &=& \indicator \left\{\frac{\sumj \indicator\{k_j\neq k_i, j\in\cG_{i}\}A_{j}}{\sumj \indicator\{k_j\neq k_i, j\in\cG_{i}\}} > \frac{1}{2}\right\}. \label{eqn::Hi_example}
\end{eqnarray}
In~\eqref{eqn::Si_example} above, $S_{i}$ indicates whether more than half of the other villages in the same sublocation as village $i$ are treated, and in~\eqref{eqn::Hi_example} above, $H_{i}$ indicates whether more than half of the nearby villages in different sublocations are treated. 

The binary exposure definitions in~\eqref{eqn::Si_example}~and~\eqref{eqn::Hi_example} above are used as a running example for exposition and align with our empirical application. More generally, the framework permits multi-level categorical exposures by replacing binary indicators with richer summaries of treatment intensity. In some settings where the exposure mapping may be continuous, we need to redefine the estimand and adopt corresponding estimation strategies, for example, coarsening the exposure into bins, applying local smoothing, or specifying an exposure--response model. We do not pursue these extensions in this paper, but discuss them and their inferential implications in Section~\ref{sec::discussion}. 

Based on the definitions in~\eqref{eqn::Si_example} and~\eqref{eqn::Hi_example}, village exposures are dependent if the corresponding villages are either in the same sublocation or geographically close. To formalize this dependence structure, we define an aggregated village network, denoted by $\cB$, together with its induced graph. We say that village $j$ is connected to village $i$ in the network if either (i) the two villages belong to the same sublocation, that is, $k_j=k_i$, or (ii) village $j$ is geographically close to village $i$, that is, $j\in \cG_i$. This relation induces a graph whose nodes are villages, with an edge between two nodes if and only if the corresponding villages are connected. Let $\bB$ denote the adjacency matrix of this graph, so that $\bB_{ij}=1$ if villages $i$ and $j$ are connected, and $\bB_{ij}=0$ otherwise. We say that villages $i$ and $j$ are within graph distance $m$ if there exists a path between them of length at most $m$ in this graph.

\paragraph{Related work and our contributions.}
Most methodological work on randomized saturation designs focuses on within-cluster interference and assumes away between-cluster interference. \cite{egger2022general} is, to our knowledge, the first empirical study in this setting to allow for between-cluster interference. Using a linear model with measures of cash transfers received by nearby villages within several radii, they find evidence suggesting the existence of between-cluster spillovers. However, their analysis is model-based and therefore relies on functional-form assumptions. More recently, \cite{Leung10102025} studies cluster-randomized trials with cross-cluster interference, explicitly allowing for interference between clusters without relying on specific models. Their main focus is on improving the estimation of direct effects and indirect effects related to treatment saturation to reduce bias when interference extends beyond cluster boundaries. By contrast, our goal is to directly define and estimate the between-cluster spillover effects themselves, providing identification and estimation results that explicitly account for between-cluster dependence.
Another related line of work is spatial interference, where researchers study interference that arises through geographic proximity or distance-based exposure \citep{papadogeorgou2022causal, giffin2023generalized, wang2025design}. In addition to geographic spillowers, however, we also consider administrative-level interference: units may interfere within the same cluster even when they are not geographically close but belong to the same administrative unit. This distinction is important in settings where social or administrative structure cannot be reduced to physical distance alone.

In this work, we make the following contributions. 
\begin{enumerate}
    \item[(i)] \textit{A design-based formulation with both within-cluster and between-cluster interference.} We extend the standard randomized saturation framework to jointly accommodate within-cluster and between-cluster spillovers. Specifically, we define an exposure mapping $d_i(\bA) = (A_i, S_i, H_i)$ that captures a unit's own treatment, treatment saturation within its cluster, and the treatment intensity among geographically proximate units in other clusters. This formulation unifies administrative and geographic interference channels in a design-based framework, where potential outcomes are treated as fixed, and treatment assignment is the only source of randomness \citep{ding2024first}. The framework operates within a finite population and does not rely on any parametric models, making the causal interpretation robust to model misspecifications.

    \item[(ii)] \textit{Causal estimands that formulate multiple direct and indirect effects.} We define a set of causal estimands for direct effects, within-cluster indirect effects, and between-cluster indirect effects. These effects are aggregated at different levels, including conditional, in-policy marginal, and policy-specific effects, to capture different aspects of the treatment effect. Together, these estimands allow researchers to examine how treatment effects vary across exposure environments and to assess whether marginalizing over unmodeled interference channels masks important heterogeneity.

    \item[(iii)] \textit{Nonparametric estimation and asymptotic theory.} We construct Horvitz--Thompson and H\'{a}jek estimators based on inverse propensity score weighting with policy-specific reweighting. We derive close-form asymptotic variance and covariance expressions (Theorem~\ref{thm:asym_var}), establish consistency and joint asymptotic normality (Theorem~\ref{thm:asymptotic}), and provide conservative variance estimators with asymptotically valid confidence intervals (Corollary~\ref{cor:ci}). The theory accommodates the complex dependence structure induced by the network and the exposure mapping under regularity conditions on potential outcomes, positivity, network degree, and order of dependence. We also briefly discussed covariate-adjusted estimators when pretreatment covariates are available. 

    \item[(iv)] \textit{Empirical re-analysis revealing economically important between-cluster spillovers.} We re-analyze the large-scale cash transfer experiment of \citet{egger2022general} in rural Kenya and obtain several new findings. First, between-sublocation spillovers are present and economically meaningful: control villages in low saturation sublocations benefit from proximity to treated villages, while treated villages in high saturation sublocations experience negative geographic spillovers. Second, treatment effects are highly heterogeneous across exposure environments, with conditional direct effects varying substantially across $(S_i,H_i)$. Third, analyses that ignore between-sublocation interference fail to detect the geographic spillover channel and can also change conclusions about some within-cluster signals, illustrating the value of accounting for the richer interference structure.
\end{enumerate}

\paragraph{Organization of the paper.}
The rest of the paper proceeds as follows. Section~\ref{sec::parameters} introduces a set of causal estimands that capture both the direct effects of the treatment and the indirect effects arising from administrative and geographic interference. Section~\ref{sec::estimation} presents point estimators for these causal estimands. Section~\ref{sec::theoretical} provides their theoretical properties, including consistency, asymptotic normality, and expressions for the asymptotic variance. Section~\ref{sec::var-est} provides variance estimators based on these theoretical results and shows that they lead to asymptotically valid confidence intervals. Section~\ref{sec::real} implements the proposed estimators and inference methods in the cash transfer study of \cite{egger2022general}. Section~\ref{sec::discussion} concludes with a discussion of future research directions. The Supplementary Material contains additional technical results.

\section{Causal estimands of interest}
\label{sec::parameters}

In this section, we formally define the causal estimands under the potential outcome framework. Following our comment earlier, we will focus on the binary choice of $(S_i, H_i)$ that appears in~\eqref{eqn::Si_example}~and~\eqref{eqn::Hi_example}, and thus $(A_i, S_i, H_i) \in \{0,1\}^3$.

\subsection{Conditional causal effects}
We first define conditional causal effects, including conditional direct effects, within-cluster indirect effects, and between-cluster indirect effects.

\paragraph{Conditional direct effects.}
We consider the direct effect of the treatment, holding $(S_{i},H_{i})$ at a fixed level $(s,h)$. Define
\begin{eqnarray*}
    \de(s,h) &=& n^{-1}\sumi Y_{i}(1,s,h) - n^{-1}\sumi Y_{i}(0,s,h).
\end{eqnarray*}
The quantity $\de(s,h)$ represents the \textit{conditional direct effect} of treatment $A_{i}$ on the outcome, while holding the exposure variables $(S_{i},H_{i})$ fixed at values $(s,h)$.

\paragraph{Conditional indirect effects.}
We consider two types of indirect effects corresponding to the two sources of interference: within-cluster and between-cluster spillover effects.

For within-cluster indirect effects, we define the \textit{within-cluster conditional indirect effect} of $S_{i}$ on the outcome as
\begin{eqnarray*}
    \wie(a,h) &=& n^{-1}\sumi Y_{i}(a,1,h) - n^{-1}\sumi Y_{i}(a,0,h).
\end{eqnarray*}
The quantity $\wie(a,h)$ captures the effect of changing the proportion of treated villages in the same sublocation (from $s = 0$ to $s = 1$) on the outcome of village $i$, while holding its own treatment fixed at $a$ and the between-cluster exposure $H_{i}$ at level $h$. 

For between-cluster indirect effects, similarly, we define the \textit{between-cluster conditional indirect effect} of $H_{i}$ on the outcome as
\begin{eqnarray*}
    \bie(a,s) &=& n^{-1}\sumi Y_{i}(a,s,1) - n^{-1}\sumi Y_{i}(a,s,0).
\end{eqnarray*}
The quantity $\bie(a,s)$ captures the effect of changing the level of between-cluster exposure from $h = 0$ to $h = 1$, with the own treatment fixed at $a$ and the within-cluster exposure $S_{i}$ fixed at level $s$. 

Under the binary exposure mapping, we can view this setup as a $2^3$ factorial experiment defined by the three binary factors $(a,s,h) \in \{0,1\}^3$. Although this perspective allows for many other causal contrasts in principle, we focus on the three conditional effects defined above because they are most relevant to our empirical motivation. See \cite{zhao2022regression} for a more general discussion on estimands in factorial experiments. 

\subsection{In-policy causal effects}
\label{subsec::in_policy_estimand}
In practice, we can report the aggregated version of the causal estimands by marginalizing over the exposure distributions induced by the implemented policy distributions. These estimands are ``in-policy'' in the sense that the
averaging distributions are generated by the same randomization design under
which the experiment is conducted.

\paragraph{Marginal direct effect.}
We consider the direct effect of the treatment, averaging over the conditional distribution of $(S_i, H_i)$. Define
\begin{eqnarray}
    \de &=& n^{-1}\sumi E\left\{Y_{i}(1,S_{i},H_{i})\mid A_i=1\right\} - n^{-1}\sumi E\left\{Y_{i}(0,S_{i},H_{i})\mid A_{i}=0\right\} \label{eq::DE} \\
    &=& n^{-1}\sumi \sum_{s=0,1} \sum_{h=0,1} \pr(S_i=s,H_i=h\mid A_i=1)Y_i(1,s,h) \notag\\
    &&- n^{-1}\sumi \sum_{s=0,1} \sum_{h=0,1} \pr(S_i=s,H_i=h\mid A_i=0)Y_i(0,s,h), \notag 
\end{eqnarray}
where the expectations in~\eqref{eq::DE} are taken with respect to the conditional distribution of $(S_i,H_i)$ given $A_i=a$, for $a=0,1$. $\de$ represents the \textit{marginal direct effect} of treatment $A_{i}$, marginalizing over the distribution of treatments for all other villages. Therefore, by definition, $\de$ depends on the treatment assignment mechanism through the conditional distributions $\pr(S_i, H_i\mid A_i)$.

\paragraph{Marginal indirect effects.}
We next define two in-policy indirect effects corresponding to the two sources of interference: within-cluster and between-cluster spillovers. For a fixed treatment status $a=0,1$, we define the \textit{marginal within-cluster indirect effect} of $S_i$ on the outcome as
\begin{eqnarray*}
    \wie(a) &=& n^{-1}\sumi E\left\{Y_{i}(a,1,H_{i})\mid (A_{i},S_{i}) = (a,1)\right\} - n^{-1}\sumi E\left\{Y_{i}(a,0,H_{i})\mid (A_{i},S_{i}) = (a,0)\right\} \\
    &=& n^{-1}\sumi \sum_{h=0,1}\pr(H_i=h\mid A_i=a,S_i=1)Y_i(a,1,h) \\
    &&- n^{-1}\sumi \sum_{h=0,1}\pr(H_i=h\mid A_i=a,S_i=0)Y_i(a,0,h).
\end{eqnarray*}
This estimand contrasts high versus low within-sublocation exposure, holding
own treatment fixed at $a$, while averaging over the between-sublocation
exposure distributions that arise under the implemented policy. As with the marginal direct effect, this definition depends on the treatment assignment mechanism.

For between-cluster indirect effects, similarly, for a fixed $a=0,1$, we define the \textit{marginal between-cluster indirect effect} of $H_{i}$ on the outcome as
\begin{eqnarray*}
    \bie(a) &=& n^{-1}\sumi E\left\{Y_{i}(a,S_{i},1)\mid (A_{i},H_{i}) = (a,1)\right\} - n^{-1}\sumi E\left\{Y_{i}(a,S_{i},0)\mid (A_{i},H_{i}) = (a,0)\right\} \\
    &=& n^{-1}\sumi \sum_{s=0,1} \pr(S_i=s\mid A_i=a,H_i=1)Y_i(a,s,1) \\
    &&- n^{-1}\sumi \sum_{s=0,1}\pr(S_i=s\mid A_i=a,H_i=0)Y_i(a,s,0).
\end{eqnarray*}
This estimand contrasts high versus low between-sublocation exposure, holding
own treatment fixed at $a$, while averaging over the within-sublocation
exposure distributions that arise under the implemented policy. As above,
$\bie(a)$ also depends on the treatment assignment mechanism.

These in-policy estimands are closely related to the marginal effects in
\citet{hudgens2008toward}. Their main advantage is policy relevance. For example, the marginal direct effect $\de$ compares the expected outcome for treated villages under the spillover environments they tend to face with the expected outcome for control villages under their corresponding spillover environments. At the same time, these estimands are not simply weighted averages of controlled effects. They compare outcomes after averaging over the spillover exposure distributions that arise under the policy, rather than holding those exposure components fixed. For example, the in-policy marginal direct effect $\de$ compares $Y_i(1,s,h)$ averaged over $\pr(S_i,H_i\mid A_i=1)$ with $Y_i(0,s,h)$ averaged over $\pr(S_i,H_i\mid A_i=0)$. If these two conditional exposure distributions differ, then $\de$ reflects not only the change in own treatment status but also the difference in spillover environments associated with treated and untreated villages under the policy. In particular, even if there is no direct effect at any fixed exposure level, so that $Y_i(1,s,h)=Y_i(0,s,h)$ for all $(s,h)$, the in-policy direct effect need not be zero when spillover effects are present and $\pr(S_i,H_i\mid A_i=1) \neq \pr(S_i,H_i\mid A_i=0)$. By contrast, under the stronger sharp null that $Y_i(a,s,h)$ is the same for all $(a,s,h)$ for each village $i$, all the in-policy effects defined above are zero.

An alternative is to define different estimands by averaging over a common reference distribution. For example, we consider the following standardized direct effect 
\[
\de_q
=
n^{-1}\sumi
\sum_{s=0,1} \sum_{h=0,1}
q_i(s,h)\{Y_i(1,s,h)-Y_i(0,s,h)\}
\]
where $q_i(s,h)$ is a prespecified reference distribution, such as the marginal distribution $\pr(S_i=s,H_i=h)$ induced by the implemented policy. Such a standardized estimand isolates the direct effect of own treatment by using the same distribution of $(S_i,H_i)$ for both treatment levels. Thus, the in-policy estimands are more directly tied to the implemented assignment policy, whereas the alternative standardized estimands are more directly interpretable as isolated direct or indirect effects. We focus on the marginal effects $\de$, $\wie(a)$, and $\bie(a)$ in this paper, while analysis of alternative estimands follows similar arguments.


\subsection{Policy-specific causal estimands}
\label{subsec::policy_specific_estimand}
In this section, we further define estimands that are specific to a hypothetical treatment assignment mechanism that could differ from the observed one, similar to \cite{hudgens2008toward}. For a specific treatment assignment policy $\psi$, define the policy-specific direct effect $\de_{\psi}$, within-cluster indirect effect $\wie_{\psi}$, and between-cluster indirect effect $\bie_{\psi}$ as follows
\begin{eqnarray*}
    \de_{\psi} &=& n^{-1}\sumi E_{\psi}\{Y_i(1,S_i,H_i)\mid A_i=1\} - n^{-1}\sumi E_{\psi}\{Y_i(0,S_i,H_i)\mid A_i=0\}, \\
    \wie_{\psi}(a) &=& n^{-1}\sumi E_{\psi}\left\{Y_{i}(a,1,H_{i})\mid (A_{i},S_{i}) = (a,1)\right\} - n^{-1}\sumi E_{\psi}\left\{Y_{i}(a,0,H_{i})\mid (A_{i},S_{i}) = (a,0)\right\}, \\
    \bie_{\psi}(a) &=& n^{-1}\sumi E_{\psi}\left\{Y_{i}(a,S_{i},1)\mid (A_{i},H_{i}) = (a,1)\right\} - n^{-1}\sumi E_{\psi}\left\{Y_{i}(a,S_{i},0)\mid (A_{i},H_{i}) = (a,0)\right\},
\end{eqnarray*}
where the expecatation $E_{\psi}(\cdot \mid \cdot)$ is taken under the distribution induced by policy $\psi$. Thus, in the definition of $\de_{\psi}$, the expectations average over the conditional distribution of $(S_i,H_i)$ given $A_i=a$ for $a=0,1$. In the definition of $\wie_{\psi}$, they average over the conditional distribution of $H_i$ given $(A_{i},S_{i}) = (0,s)$ for $s=0,1$, and similarly for $\bie_{\psi}$. In particular, if we take $\psi$ as the treatment policy implemented in the study, the above effects recover the in-policy causal estimands we defined in the previous section. 

For two policies $\psi_1$, $\psi_2$, we can therefore compare their direct effect and indirect effects by the contrasts $\de_{\psi_1} - \de_{\psi_2}$, $\wie_{\psi_1} - \wie_{\psi_2}$, and $\bie_{\psi_1} - \bie_{\psi_2}$, respectively.

\section{Estimation by inverse propensity score weighting}
\label{sec::estimation}
In this section, we provide inverse propensity score weighting estimators for the causal estimands defined in Section~\ref{sec::parameters}. We construct estimators for average potential outcomes and thus for the conditional effects, and then introduce more policy-specific causal effect estimators.

We start by introducing the propensity score. For a given exposure mapping level $(a,s,h)$, define the propensity score as $\pscore{i}(a,s,h)=\pr(A_{i}=a, S_{i}=s, H_{i}=h)$. These propensity scores are determined by both the experimental design and the village network structure. In principle, they depend deterministically on the treatment assignment mechanism and definition of the exposure mapping, but analytically calculating them can be infeasible as the treatment space grows. A practical alternative is to approximate these probabilities using Monte Carlo simulation: we draw $A$ from the treatment assignment mechanism, and calculate $A, S, H$ according to \eqref{eqn::Si_example} and \eqref{eqn::Hi_example}. Then we can estimate $\pi_i$ by the empirical frequency of simulated exposures. We next introduce the inverse propensity score weighted estimators in the following subsections.

\subsection{Averages of the potential outcomes and conditional effects}
As a building block, we construct an estimator for the average potential outcome,
\begin{eqnarray*}
    \bar{Y}(a,s,h) &=& n^{-1}\sumi Y_{i}(a,s,h),
\end{eqnarray*}
for given $(a,s,h)$. Consider the Horvitz--Thompson estimator
\begin{eqnarray*}
    \hat{Y}^{\horthom}(a,s,h) &=& n^{-1}\sumi\frac{ \bbI_{i}(a,s,h)}{ \pscore{i}(a,s,h)}Y_{i},
\end{eqnarray*}
and the H\'{a}jek estimator
\begin{eqnarray*}
    \hat{Y}^{\haj}(a,s,h) &=& n^{-1}\sumi\frac{ \bbI_{i}(a,s,h) }{ \pscore{i}(a,s,h)}Y_{i} \Big / n^{-1}\sumi\frac{ \bbI_{i}(a,s,h)}{ \pscore{i}(a,s,h)},
\end{eqnarray*}
where $\bbI_{i}(a,s,h)=\indicator\{A_{i}=a,S_{i}=s,H_{i}=h\}$ is the indicator for the exposure mapping level $(a,s,h)$. Under randomization, the estimator $\hat{Y}^{\horthom}(a,s,h)$ is unbiased to $\bar{Y}(a,s,h)$. The H\'{a}jek estimator $\hat{Y}^{\haj}(a,s,h)$ is not unbiased in finite samples but is consistent to $\hat{Y}^{\haj}(a,s,h)$ and generally has more stable finite sample performance. Therefore, the corresponding Horvitz--Thompson estimators are unbiased, and the H\'{a}jek estimators are consistent.

We then propose the following estimators for the conditional direct and indirect effects defined in Section~\ref{sec::parameters}. For $*\in\{\horthom,\haj\}$, construct
\begin{eqnarray*}
    \hat{\de}^{*}(s,h) &=& \hat{Y}^{*}(1,s,h) - \hat{Y}^{*}(0,s,h), \\
    \hat{\wie}^{*}(a,h) &=& \hat{Y}^{*}(a,1,h) - \hat{Y}^{*}(a,0,h), \\
    \hat{\bie}^{*}(a,s) &=& \hat{Y}^{*}(a,s,1) - \hat{Y}^{*}(a,s,0).
\end{eqnarray*}

\subsection{In-policy and policy-specific causal effects}

We next propose to use a policy-specific re-weighting to estimate the policy-specific effects. Under different policies, the joint distribution of $(A_i,S_i,H_i)$ will, in general, differ. For a given policy of interest, the policy-specific causal effects are defined as marginal expectations of the potential outcomes, where the marginalization is taken with respect to the policy-induced distribution of $(A_i,S_i,H_i)$, as defined in Sections~\ref{subsec::in_policy_estimand}~and~\ref{subsec::policy_specific_estimand}. Each estimand involves a reweighted average of potential outcomes $\bar{Y}(a,s,h;\Gamma) = n^{-1} \sumi \gamma_i(a,s,h)Y_i(a,s,h)$  indexed by $\Gamma = \{\gamma_i(a,s,h): i = 1,\dots,n; a,s,h = 0,1\}$, where the weights depend on the distribution of $(A_i,S_i,H_i)$ under a specific policy. Accordingly, to estimate these effects, we construct a class of reweighting estimators that account for the policy-induced exposure distribution.
More concretely, for a given set of weights $\gamma_i(a,s,h)$, define the following class of Horvitz--Thompson estimators:
\begin{eqnarray*}
    \hat{Y}^{\horthom}(a,s,h;\Gamma) &=& n^{-1}\sumi\frac{ \bbI_{i}(a,s,h)\gamma_i(a,s,h)}{ \pscore{i}(a,s,h)}Y_{i},
\end{eqnarray*}
and the corresponding class of H\'{a}jek estimators:
\begin{eqnarray*}
    \hat{Y}^{\haj}(a,s,h;\Gamma) &=& n^{-1}\sumi \gamma_i(a,s,h)\cdot n^{-1}\sumi\frac{ \bbI_{i}(a,s,h) \gamma_i(a,s,h)}{ \pscore{i}(a,s,h)}Y_{i} \Big / n^{-1}\sumi\frac{ \bbI_{i}(a,s,h) \gamma_i(a,s,h)}{ \pscore{i}(a,s,h)}. 
\end{eqnarray*}
These estimators form a general framework for constructing estimators of the policy-specific direct and indirect effects, as introduced in Section~\ref{sec::parameters}. 

To operationalize the reweighting idea, we specify weight functions for each type of causal effect. Let $\Gamma_{\psi}^{*} = \{\gamma_{i,\psi}^{*}(a,s,h): i = 1,\dots,n; a,s,h = 0,1\}$ with $* = \{\text{\de, \wie, \bie}\}$ denote the corresponding classes of weighting functions for the direct, within-sublocation indirect, and between-sublocation indirect effects, respectively. The elements of these classes are defined as
\begin{eqnarray*}
    \gamma_{i,\psi}^{\de}(a,s,h) &=& \pr_{\psi}(S_i=s,H_i=h\mid A_i=a), \\
    \gamma_{i,\psi}^{\wie}(a,s,h) &=& \pr_{\psi}(H_i=h\mid A_i=a, S_i=s), \\
    \gamma_{i,\psi}^{\bie}(a,s,h) &=& \pr_{\psi}(S_i=s\mid A_i=a, H_i=h).
\end{eqnarray*}
These probabilities describe the distribution of $(S_i,H_i)$ under a specific policy $\psi$, which serve as reweighting terms in our estimators. For either the Horvitz--Thompson or H\'{a}jek estimator ($*\in\{\horthom,\haj\}$), we then define the corresponding estimators for the policy-specific causal effects as:
\begin{eqnarray*}
    \hat{\de}_{\psi}^{*} &=& \sum_{s,h=0,1} \{\hat{Y}^{*}(1,s,h;\Gamma_{\psi}^{\de})  - \hat{Y}^{*}(0,s,h;\Gamma_{\psi}^{\de})\}, \\
    \hat{\wie}_{\psi}^{*}(a) &=& \sum_{h=0,1} \{\hat{Y}^{*}(a,1,h;\Gamma_{\psi}^{\wie}) - \hat{Y}^{*}(a,0,h;\Gamma_{\psi}^{\wie})\}, \\
    \hat{\bie}_{\psi}^{*}(a) &=& \sum_{s=0,1} \{\hat{Y}^{*}(a,s,1;\Gamma_{\psi}^{\bie}) - \hat{Y}^{*}(a,s,0;\Gamma_{\psi}^{\bie})\}.
\end{eqnarray*}
In the definition of $\hat{\de}_{\psi}^{*}$, the first term in the summation uses the subset of weights in $\Gamma^{\de}_\psi$ given $A=1$, and the second term uses the weights given $A=0$. Similarly, each summand in $\hat{\wie}_{\psi}^{*}$ and $\hat{\bie}_{\psi}^{*}$ use only a subset of the weights. 

As a special case, if we take $\psi$ as the treatment policy actually implemented in the real study, these expressions yield the in-policy estimators for the marginal causal effects. Similar to before, we omit the subscript $\psi$ for quantities corresponding to the in-policy treatment for convenience. 

\section{Theoretical properties}
\label{sec::theoretical}
In this section, we formally establish the theoretical properties of the proposed estimators. We introduce the required regularity conditions, provide the asymptotic variances and covariances of the estimators, and prove their consistency and asymptotic normality. To improve readability, readers may first consult Sections~\ref{sec:asymptotic-var} and~\ref{sec:normality} for the main theoretical results, and then return to Section~\ref{sec:condition} for the underlying conditions.

\subsection{Regularity conditions}\label{sec:condition}

In order to provide the explicit form of asymptotic covariances and to establish the asymptotic properties of the proposed estimators, we first introduce several regularity conditions in this subsection.

\begin{assumption}[Bounded potential outcomes]\label{ass:bounded}
There exists a constant $C_Y < \infty$ such that $|Y_i(a,s,h)| \leq C_Y$ for all $i = 1, \ldots, n$ and all $(a,s,h) \in \{0,1\}^3$. 
\end{assumption}
Assumption~\ref{ass:bounded} can be relaxed to some moment conditions. We keep the boundedness assumption to simplify the presentation. 

\begin{assumption}[Positivity of exposure probabilities]\label{ass:positivity}
There exist constants $0 < \underline{c}_{\pi} < \overline{c}_{\pi} < 1$ such that $\underline{c}_{\pi} \leq \pscore{i}(a,s,h) \leq \overline{c}_{\pi}$ for all $i = 1, \ldots, n$ and all $(a,s,h) \in \{0,1\}^3$. 
\end{assumption}

Assumption \ref{ass:positivity} extends the classical positivity assumption to the interference setting, and is a feature of both the experimental design and the exposure mapping. It requires that every exposure combination has a non-negligible probability of occurring. This ensures that the inverse propensity score weights remain well-behaved and do not explode, which is essential both for identification of the causal estimands and for controlling the variances of the inverse propensity score weighted estimators. In practice, positivity requires more careful checking when the exposure mapping has a large support relative to the effective sample size or when the village network degree is high or highly heterogeneous.

In applied work, we recommend overlap diagnostics before the effect estimation. These may include tabulating sample counts and distribution summaries by exposure cells, examining lower-tail summaries of the exposure probabilities $\pi_i(a,s,h)$, or assessing the extent to which inverse propensity score weights are concentrated or take extreme values.

It is helpful to address the positivity violation at both the design and analysis stages. At the design stage, we can try to avoid assignment schemes that create nearly deterministic exposure patterns. At the analysis stage, we can coarsen the exposure mapping, for example, by using low/medium/high categories, or use smaller nearby village radii. In such cases, inference should focus on contrasts with enough support in the data.

Next, we impose the following sparsity condition on the village network $\cB$.
\begin{assumption}[Bounded network degree]\label{ass:degree}
There exists a finite constant $\Delta < \infty$ such that each unit has at most $\Delta$ neighbors: 
$\max_{i=1,\dots,n}(|\{j: k_j = k_i\} \cup \cG_{i}|) \le \Delta$. 
\end{assumption}
Assumption~\ref{ass:degree} is a sparsity condition on the interference network that restricts the amount of interference in the network. Similar restrictions have also been imposed by others \citep{aronow2017estimating, li2022random}. It requires that each unit have only a limited number of neighbors. In the cash transfer study, this means that for each unit $i$, the total number of villages that either lie in the same sublocation as $i$ or are geographically nearby is uniformly bounded by $\Delta$. This condition is important for guaranteeing a stable asymptotic distribution for the estimators, as it prevents the dependence structure from becoming too dense when $n$ grows. Dense networks can violate the classical dependency-graph conditions required for central limit theorems, and in such regimes, we would need either additional restrictions on how quickly degrees are allowed to grow or an alternative asymptotic framework designed for dense networks. Another possible violation of Assumption~\ref{ass:degree} is a sparse network with finite average degree but unbounded maximum degree, such as power-law degree distribution and Erd\H{o}s--R\'{e}nyi networks. Our theoretical results, therefore, focus on settings with uniformly bounded degree. If $\Delta$ diverges slowly, our theory can still hold under some technical improvement. 
When $\Delta$ diverges fast with dense networks, we need a different theory. We do not pursue these technical directions in the current paper. In our empirical application with $n = 653$ villages, Assumption \ref{ass:degree} is plausible, with $\Delta \approx 20$, corresponding to the maximum number of villages in any sublocation plus nearby villages within 4 km. 

\begin{assumption}[Bounded order of dependence]\label{ass:dependence}
The definition of exposure mappings $d_i(\ba)$ satisfies a bounded dependence condition: there exists a positive integer $m$, such that the form $d_i(\ba)$ only depends on the treatment status $a_j$ for village $j$ within distance $m$ of village $i$ on the village network graph $\cB$. 
\end{assumption}
Assumption \ref{ass:dependence} describes how far dependence can propagate in the interference structure. It implies that units sufficiently far apart in the network behave independently, which is another key requirement for applying the dependency graph central limit theorems. Intuitively, the dependence induced by randomization does not extend indefinitely: units outside each other's $2m$-step neighborhoods cannot influence each other's exposure conditions. In our empirical setting, the bounded dependence assumption holds with $m = 1$ because each village's exposure depends only on its own sublocation and immediate neighbors, and any units separated by two or more steps behave independently conditional on the randomization. 

Under these assumptions, we provide the asymptotic covariances and establish the asymptotic results for our policy-specific estimators in the following two subsections.

\subsection{Asymptotic variance}\label{sec:asymptotic-var}
For readability, we first introduce the asymptotic variances of the reweighted estimators in this section before formally proving the asymptotic distribution results. These asymptotic variances will serve as building blocks for the asymptotic distribution results in the next subsection. Let $\bGamma(a,s,h)=\diag\{\gamma_i(a,s,h)\}$ denote the $n\times n$ diagonal matrix with diagonal terms equal to the weights $\gamma_i(a,s,h)$ for $i=1,\ldots,n$. To characterize dependence induced by the randomization design, we further define two $n\times n$ matrices, $\bLambda(a,s,h)$ and $\bLambda(a,s,h;\ashprime)$, which will be crucial in defining the asymptotic variances:
\begin{itemize}
    \item For $(a,s,h)$, the $(i,j)$-th entries are: 
    \begin{eqnarray*} 
        \bLambda_{ij}(a,s,h) = 
        \indicator\{j=i\} \frac{1 - \pscore{i}(a, s, h)}{\pscore{i}(a, s, h)} + \indicator\{j\neq i\}\frac{\pscore{ij}(a, s, h; a, s, h) - \pscore{i}(a, s, h)\pscore{j}(a, s, h)}{\pscore{i}(a, s, h)\pscore{j}(a, s, h)}, 
    \end{eqnarray*}
    where $\pscore{ij}(a,s,h;a,s,h)=\pr(A_i=A_j=a,S_i=S_j=s,H_i=H_j=h)$ is the second-order inclusion probabilities for a given exposure mapping level $(a,s,h)$.

    \item For $(a,s,h) \neq (\ashprime)$, the $(i,j)$-th entries are: 
    \begin{eqnarray*} 
         \bLambda_{ij}(a,s,h;\ashprime) = 
         -\indicator\{j=i\}+ \indicator\{j\neq i\}\frac{\pscore{ij}(a, s, h; \ashprime) - \pscore{i}(a, s, h)\pscore{j}(\ashprime)}{\pscore{i}(a, s, h)\pscore{j}(\ashprime)},
    \end{eqnarray*}
    where $\pscore{ij}(a, s, h; \ashprime)=\pr(A_i=a,S_i=s,H_i=h,A_j=a^{\prime},S_j=s^{\prime},H_j=h^{\prime})$ is the second-order inclusion probabilities for a pair of different exposure mapping levels $(a,s,h) \neq (\ashprime)$.
\end{itemize}

In practice, these second-order inclusion probabilities $ \pi_{ij}(a,s,h;\ashprime)$ can be simulated given a known randomization policy and network structure. See \cite{aronow2017estimating} for further discussion. 

In addition, the asymptotic variance depends on whether the potential outcomes are centered, and if so, around what quantity. The Horvitz--Thompson estimator uses the raw potential outcomes and therefore has a variance expression directly in terms of $Y_i(a,s,h)$, so we define $Y_i^{\horthom}(a,s,h)=Y_i(a,s,h)$. In contrast, the H\'{a}jek estimator normalizes the weights by an estimated denominator, which induces additional dependence across units. It is therefore convenient to rewrite the H\'{a}jek estimator as an inverse propensity score weighting estimator applied to centered potential outcomes. Accordingly, for the H\'{a}jek form we define
\begin{eqnarray*}
    Y_i^{\haj}(a,s,h) = 
    Y_i(a,s,h) - \frac{n^{-1}\sumi \gamma_i(a,s,h)Y_i(a,s,h)}{n^{-1}\sumi\gamma_i(a,s,h)}.
\end{eqnarray*}
Let $\bY^{*}(a,s,h)=(Y_1^{*}(a,s,h),\ldots,Y_n^{*}(a,s,h))^{\T}$ denote the vector of potential outcomes for all units for $*\in\{\horthom,\haj\}$. 

\begin{theorem}\label{thm:asym_var}
For a reweighting regime $\Gamma$ and $*\in\{\horthom,\haj\}$, under Assumptions~\ref{ass:bounded}--\ref{ass:dependence}, the asymptotic variance of $\hat{Y}^{*}(a,s,h;\Gamma)$ at a fixed treatment and exposure mapping level $(a,s,h)$ is
\begin{equation*}
    \avar\{\hat{Y}^{*}(a,s,h;\Gamma)\} = n^{-2} \bY^{*}(a,s,h)^{\T} \bGamma(a,s,h) \bLambda(a,s,h) \bGamma(a,s,h) \bY^{*}(a,s,h), 
\end{equation*}
and the asymptotic covariance for a given pair $(a,s,h)\neq (a^{\prime},s^{\prime},h^{\prime})$ is
\begin{equation*}
    \acov\{\hat{Y}^{*}(a,s,h;\Gamma), \hat{Y}^{*}(a',s',h';\Gamma)\} = n^{-2} \bY^*(a,s,h)^{\T} \bGamma(a,s,h)\bLambda(a,s,h;\ashprime) \bGamma(a',s',h')\bY^*(\ashprime),
\end{equation*}
where $\avar(\cdot)$ and $\acov(\cdot, \cdot)$ denote the asymptotic variance and asymptotic covariance, respectively.
\end{theorem}

Theorem~\ref{thm:asym_var} provides a compact quadratic form representation of the asymptotic variance under a general reweighting scheme $\Gamma$. This asymptotic covariance depends on the joint distribution of the potential outcomes across different exposure levels, $\{Y_i(a,s,h)\}_{a,s,h\in\{0,1\}}$, which is not directly observable. As a result, consistent variance estimation is generally not feasible. We discuss this issue further in Section~\ref{sec::var-est}.

\subsection{Consistency and asymptotic normality}\label{sec:normality}

We next establish the asymptotic properties of the proposed policy-specific estimators. Following the framework of \citet{aronow2017estimating} for inverse propensity score weighting estimators under network interference and \citet{chen2004normal} for central limit theorems under network dependence, we establish consistency and asymptotic normality. To show this, we assume one more mild regularity condition to guarantee the non-degeneracy of the asymptotic variances. In the proof of Theorem \ref{thm:asym_var}, we have shown that the asymptotic variance of the estimators is of order $O(n^{-1})$, but did not show a lower bound on this order. For extreme specifications of potential outcomes, the asymptotic variance can be arbitrarily small (for example, if all potential outcomes are zero). To avoid these cases, we assume that the asymptotic variances are non-degenerate, i.e., are of the exact order of $n^{-1}$, denoted as $\Theta(n^{-1})$.

\begin{theorem}[Consistency and asymptotic normality]\label{thm:asymptotic}
Suppose Assumptions~\ref{ass:bounded}--\ref{ass:dependence} hold. Let $$\hat{\bY}^{*}_{\Gamma} = (\hat{Y}^{*}(a,s,h;\Gamma))^\T_{(a,s,h)\in\{0,1\}^3}, \quad \bar{\bY}_{\Gamma} = (\bar{Y}(a,s,h;\Gamma))^\T_{(a,s,h)\in\{0,1\}^3}$$ 
denote the vectors of estimators and average weighted potential outcomes across all exposure combinations under the weighting regime $\Gamma$, for $*\in\{\horthom, \haj\}$. Also, assume that the asymptotic variances of the estimators are non-degenerate: 
\begin{align}\label{eqn:non-degeneracy}
    \avar(\hat{\bY}^{*}_{\Gamma}) = \Theta(n^{-1}). 
\end{align}
Then as $n \to \infty$:
\begin{enumerate}
    \item \textbf{Consistency:} For any weight function $\Gamma$ and any exposure $(a,s,h)$, we have for $*\in\{\horthom, \haj\}$,
    \begin{equation*}
    \hat{\bY}^{*}_\Gamma -  \bar{\bY}_\Gamma \xrightarrow{p} 0.
    \end{equation*}
    Consequently, all proposed estimators for the causal effects, including the conditional effects, the in-policy marginal effects, and the policy-specific effects, are consistent for their respective population estimands.
    
    \item \textbf{Asymptotic normality:} The joint vector of estimators satisfies
    \begin{equation*}
    \acov(\hat{\bY}^{*}_{\Gamma})^{-1/2}\left(\hat{\bY}^{*}_{\Gamma} - \bar{\bY}_\Gamma\right) \xrightarrow{d} N\left(0, I_8\right)
    \end{equation*}
    for $*\in\{\horthom, \haj\}$, where $I_8$ is the identity matrix of dimension 8.
\end{enumerate}
\end{theorem}

\section{Variance estimation and inference}
\label{sec::var-est}
In this section, we develop a variance estimation and confidence interval construction framework that yields asymptotically valid inference for the causal effects introduced above. We also propose a covariate adjustment strategy which can potentially help improve efficiency through variance reduction.

\subsection{Variance estimation}
\label{sec::var-est-subsection}
The asymptotic covariances in Theorem~\ref{thm:asym_var} require knowledge of the true potential outcomes $Y_i(a,s,h)$, which are not directly observable. In this section, we propose conservative variance estimators that rely only on observed data. 

Define the aggregated vector $\hat{\bY}^{*} = (\hat{Y}_1^{*}, \ldots, \hat{Y}_n^{*})^{\T}$, where for $*\in\{\horthom,\haj\}$,  $\hat{Y}_i^{\horthom} = \bbI_{i}(a,s,h) Y_i / \pscore{i}(a,s,h)$, and $\hat{Y}_i^{\haj} = \bbI_{i}(a,s,h) \tilde{Y}_i / \pscore{i}(a,s,h)$ where
$$\tilde{Y}_i = Y_i - \frac{\hat{Y}^{\haj}(a,s,h;\Gamma)}{n^{-1}\sumi \gamma_i(a,s,h)}.$$
Here the centering for the H\'{a}jek estimator involves a summation of weights $\gamma_i(a,s,h)$ in the denominator for targeting a general weighting scheme, while Horvitz--Thompson does not involve such centering or scaling. Next, define the $n\times n$ matrix $\bOmega(a,s,h)$ with $(i,j)$-th entries: 
$$\bOmega_{ij}(a,s,h) = \frac{\pscore{ij}(a,s,h;a,s,h)-\pscore{i}(a,s,h)\pscore{j}(a,s,h)}{\pscore{ij}(a,s,h;a,s,h)}.$$
Then we construct the following variance estimator for $\avar\{\hat{Y}^{*}(a,s,h;\Gamma)\}$ for $*\in\{\horthom,\haj\}$:
\begin{eqnarray*}
    \hat{\textup{se}}^2\{\hat{Y}^{*}(a,s,h;\Gamma)\} &=& n^{-2} (\hat{\bY}^{*})^{\T} \bGamma(a,s,h) \bOmega(a,s,h) \bGamma(a,s,h) \hat{\bY}^{*}.
\end{eqnarray*}

For conditional causal effects such as $\hat{\de}^{*}(s,h) = \hat{Y}^{*}(1,s,h) - \hat{Y}^{*}(0,s,h)$, the true asymptotic variance includes the covariance term that depends on the joint of different potential outcomes and is not identified. Decomposing the variance, 
\begin{eqnarray*}
    \var\{\hat{\de}^{*}(s,h)\} &=& \var\{\hat{Y}^{*}(1,s,h)\} + \var\{\hat{Y}^{*}(0,s,h)\} - 2\cov\{\hat{Y}^{*}(1,s,h), \hat{Y}^{*}(0,s,h)\},
\end{eqnarray*}
reveals that the two variance components are identifiable but the covariance is not. For valid inference, we obtain an upper bound that is identifiable based on the observed data for the covariance term. We use the fact that $\cov^2\{\hat{Y}^{*}(1,s,h), \hat{Y}^{*}(0,s,h)\}\leq \var\{\hat{Y}^{*}(1,s,h)\} \var\{\hat{Y}^{*}(0,s,h)\}$ guaranteed by the Cauchy--Schwarz inequality, with equality holding when the two estimators $\hat{Y}^{*}(1,s,h)$ and $\hat{Y}^{*}(0,s,h)$ are perfectly correlated.
This motivates the following conservative variance estimator for $*\in\{\horthom,\haj\}$,
\begin{equation*}
\hat\var\{\hat{\de}^{*}(s,h)\} = \left[\hat{\textup{se}}\{\hat{Y}^{*}(1,s,h)\} + \hat{\textup{se}}\{\hat{Y}^{*}(0,s,h)\} \right]^2.
\end{equation*}
We can similarly construct conservative variance estimators for other conditional causal effects. Such a Cauchy--Schwarz type of variance estimator has been proposed and studied before in design-based inference. For example, \cite{neyman1923application} recommended choosing the Cauchy--Schwarz-based variance estimator as a conservative variance approximation to the truth in the non-interference setting, which is the sharpest given second-order moments of the potential outcomes. Problem 4.5 of \cite{ding2024first} also discusses the properties of this variance estimator. Our proposal here extends a similar construction to the interference case.

For policy-specific causal effects $\hat{\de}{}^{*} = \sum_{s,h} \{\hat{Y}^{*}(1,s,h;\Gamma^{\de}_{\psi}) - \hat{Y}^{*}(0,s,h;\Gamma^{\de}_{\psi})\}$, we use $\hat\var(\hat{\de}^{*}) = [\sum_{a=0,1}\sum_{s=0,1}\sum_{h=0,1}\hat{\textup{se}}\{\hat{Y}^{*}(a,s,h;\Gamma^{\de}_{\psi})\}]^2$ as a variance estimator, by applying the Cauchy--Schwarz inequality to all covariance terms $\cov\{\hat{Y}^{*}(a,s,h;\Gamma^{\de}_{\psi}), \hat{Y}^{*}(a',s',h';\Gamma^{\de}_{\psi})\}$. We can similarly construct variance estimators for other policy-specific causal effects. These estimators are generally not consistent for variance estimation, but they are conservative by construction so that they guarantee valid asymptotic coverage for inference. 

\subsection{Inference}
In this section, we provide asymptotic results for the variance estimator to support valid inference.
\begin{corollary}[Asymptotic validity of confidence intervals]\label{cor:ci}
Under Assumptions~\ref{ass:bounded}--\ref{ass:dependence}, and further assuming that the second-order inclusion probabilities are positive: $\pscore{ij}(a,s,h;a,s,h) > 0$, and the asymptotic variances of the estimators are non-degenerate: $\avar(\hat{\bY}^{\horthom}_{\Gamma}) = \Theta(n^{-1})$, $\avar(\hat{\bY}^{\haj}_{\Gamma}) = \Theta(n^{-1})$, we have
\begin{enumerate} 
    \item  The variance estimators $\hat{\var}\{\hat{Y}^{*}(a,s,h;\Gamma)\}$ for $* \in \{\horthom, \haj\}$ is consistent: 
    \begin{equation*}
    n\left[\hat{\var}\{\hat{Y}^{*}(a,s,h;\Gamma)\} - \avar\{\hat{Y}^{*}(a,s,h;\Gamma)\}\right] = o_p(1),
    \end{equation*}
    and therefore for a single exposure configuration $(a,s,h)$, the confidence interval
    \begin{equation*}
    \hat{Y}^{*}(a,s,h;\Gamma) \pm z_{\alpha/2} \cdot \widehat{\textup{se}}\{\hat{Y}^{*}(a,s,h;\Gamma)\}
    \end{equation*}
    achieves asymptotic coverage rate $1 - \alpha$ for $* \in \{\horthom, \haj\}$, where $z_{\alpha/2}$ denotes the upper $\alpha/2$ quantile of the standard normal distribution.
    
    \item For conditional causal effects and policy-specific causal effects, the conservative variance estimator based on the Cauchy--Schwarz bound provides asymptotically valid confidence intervals with asymptotic coverage rate at least $1 - \alpha$. In general, the actual coverage may exceed the nominal level due to the conservativeness of the bound.
\end{enumerate}
\end{corollary}

\begin{remark}
    If the second-order inclusion probability $\pscore{ij}(a,s,h;a,s,h) = 0$ for some pair $(i,j)$, the variance of $\hat{Y}^{*}(a,s,h)$ cannot be consistently estimated either. We can construct conservative estimators for the variance term using the idea of  \cite{aronow2017estimating}.
\end{remark}

\subsection{Covariate-adjusted estimator}

In this section, we examine how pre-treatment covariates can be incorporated to improve asymptotic efficiency when they are predictive of the potential outcomes. Regression-assisted estimation strategies of this kind have been studied extensively in the design-based inference literature. For example, \cite{lin2013agnostic} proposed an interacted regression adjustment to improve efficiency in completely randomized experiments, \cite{fogarty2018regression} developed a regression adjustment strategy for matched-pairs experiments, and \cite{su2021model} studied regression adjustment for cluster-randomized experiments. Motivated by this line of work, we propose the following covariate-adjusted estimator.

Let $X_i$ denote the vector of centered pre-treatment covariates for unit $i$, including a constant term in the first position. For each exposure configuration $(a,s,h)\in\{0,1\}^3$, we define a covariate-adjusted estimator of the average potential outcome $\bar{Y}(a,s,h;\Gamma)$ as
\begin{eqnarray*}
    \hat{Y}^{\textup{ca}}(a,s,h;\Gamma) &=& n^{-1}\sumi \left[\frac{\indicator_i(a,s,h)\gamma_i(a,s,h)}{\pscore{i}(a,s,h)}\bigl\{Y_i-\hat\beta(a,s,h)^{\T}X_i\bigr\} + \gamma_i(a,s,h)\hat\beta(a,s,h)^{\T}X_i\right],
\end{eqnarray*}
where $\hat\beta(a,s,h)$ is the ordinary least squares coefficient obtained by regressing $Y_i$ on $X_i$ within the subsample with $(A_i,S_i,H_i)=(a,s,h)$. The estimator $\hat{Y}^{\textup{ca}}(a,s,h)$ combines inverse propensity score weighting with regression adjustment. This type of estimator is standard in the literature \citep{scharfstein1999adjusting, bang2005doubly}. A detailed discussion is also available in Chapter 12 of \citet{ding2024first}. This augmented inverse propensity score weighting estimator intuitively improves the efficiency if the covariates are predictive of the outcome. However, we do not pursue the rigorous theoretical guarantee in the current work. Developing rigorous theory would be an interesting future direction.

For conditional causal effects, we estimate them by taking differences between the corresponding $\hat{Y}^{\textup{ca}}(a,s,h)$ estimators. For marginal causal effects, we form weighted averages of the $\hat{Y}^{\textup{ca}}(a,s,h)$ estimators and then take the corresponding contrasts, where the weights are given by the policy-induced probabilities.


\section{Real data analysis}\label{sec::real}

\subsection{Overview of analysis paradigms and results}
We apply our methodology to re-analyze the cash transfer experiment studied in \citet{egger2022general}. The experiment assigned villages to treatment using a two-stage randomized saturation design across $653$ villages in 155 sublocations, resulting in 328 treated and 325 control villages. Our empirical analysis focuses on four village-level enterprise outcomes measured at endline: winsorized average profit, revenue, total cost, and wage bill. We apply our proposed methodology to quantify direct and indirect effects.

Our re-analysis differs from \citet{egger2022general} in two main ways. First, their original study primarily used regression-based analyses to estimate direct effects and within-sublocation spillovers generated by variation in treatment saturation. Second, their analysis proposed models to capture spillover effects but did not decompose the spillovers into within- and between-cluster spillovers. In contrast, our framework introduces a separate geographic exposure variable $H_i$ and treats between-sublocation interference as an object of direct interest. This allows us to define and estimate direct effects, within-sublocation spillovers, and between-sublocation spillovers under a unified design-based framework. 

Overall, three messages emerge from our re-analysis. First, spillovers are not only within sublocations: nearby treated villages in other sublocations also affect outcomes. Second, treatment effects are highly heterogeneous across exposure environments, so averaging over exposure can mask substantively important patterns. Third, the within- and between-sublocation channels interact: favorable local spillovers under high sublocation saturation can be weakened or reversed when geographic exposure is also high.

A key feature of this setting is that interference may arise through two distinct channels: (i) villages in the same sublocation share administrative and economic ties, and (ii) nearby villages in different sublocations may also influence each other through geographic proximity. To capture these two sources of interference, we construct two binary exposure variables $S_i$ and $H_i$ following  Section~\ref{sec::parameters}. For $S_i$, we summarize treatment saturation within $i$’s sublocation. For $H_i$, we build a geographic network using distance data, identifying for each village $i$ up to three nearest villages \emph{outside} its sublocation and within 4 km.

For practical guidance on choosing the geographic exposure definition, we suggest using three considerations when selecting the distance threshold and the nearby village count $k$. First, substantive knowledge should be used: the scale should reflect where cross-sublocation interference is likely to occur, given local market integration, mobility, and social or economic connections. Second, design and measurement constraints should be considered, because geocoding precision, administrative boundaries, and data quality may limit how finely neighborhoods can be defined. Third, diagnostic feasibility is also important: candidate definitions can be assessed by checking the overlap of the induced exposure states and the stability of the estimated effects over a scientifically plausible range. Following this principle, we report results under several neighborhood definitions in Section~\ref{sec::additional_empirical_results} of the Supplementary Material.

\subsection{Propensity score of the exposure mapping}

Because both $S_i$ and $H_i$ depend on the treatment assignments of multiple nearby villages, the propensity scores $\pscore{i}(a,s,h)$ are not available in closed form due to the complex dependency between $(A_i, S_i, H_i)$ induced by the network structure. We therefore estimate the propensity scores $\pi_i(a,s,h)$ and the second-order inclusion probabilities $\pi_{ij}(a,s,h; a',s',h')$ using $100{,}000$ Monte Carlo draws following \citet{aronow2017estimating}. 

Computationally, the Monte Carlo cost scales approximately linearly with the number of draws and the per-draw cost of generating assignments and recomputing exposures. In sparse networks, computation can be reduced by exploiting dependence structure: many unit pairs are effectively independent and need not be treated as fully coupled in second-order calculations. From an implementation perspective, the Monte Carlo step is naturally parallel across draws and can be distributed across cores or machines. Memory use can be controlled by streaming sufficient statistics rather than storing all simulated assignments. In practice, we recommend starting from a moderate value for the number of draws and increasing it until probability summaries, weight diagnostics, and target estimates stabilize. 

Table~\ref{tab:propensity_scores} reports summary statistics for the estimated propensity scores across the 653 villages in our sample.

\begin{table}[htbp]
\centering
\caption{Summary statistics of estimated propensity scores}
\label{tab:propensity_scores}
\begin{tabular}{cccc}
\toprule
Exposure & Mean & Std & Median \\
\midrule
$(A_i=0,S_i=0,H_i=0)$ & 0.184 & 0.175 & 0.137 \\
$(A_i=0,S_i=0,H_i=1)$ & 0.143 & 0.141 & 0.099 \\
$(A_i=0,S_i=1,H_i=0)$ & 0.088 & 0.062 & 0.071 \\
$(A_i=0,S_i=1,H_i=1)$ & 0.086 & 0.068 & 0.068 \\
$(A_i=1,S_i=0,H_i=0)$ & 0.118 & 0.078 & 0.088 \\
$(A_i=1,S_i=0,H_i=1)$ & 0.098 & 0.065 & 0.078 \\
$(A_i=1,S_i=1,H_i=0)$ & 0.139 & 0.143 & 0.069 \\
$(A_i=1,S_i=1,H_i=1)$ & 0.144 & 0.154 & 0.051 \\
\midrule
\multicolumn{4}{l}{\textit{Marginal probabilities}} \\
$\pr(A_i = 1)$ & 0.499 & 0.167 & 0.500 \\
$\pr(S_i = 1)$ & 0.458 & 0.342 & 0.500 \\
$\pr(H_i = 1)$ & 0.471 & 0.224 & 0.500 \\
\bottomrule
\end{tabular}
\begin{flushleft}
\footnotesize
\textit{Notes:} Propensity scores are estimated using 100{,}000 Monte Carlo draws. Marginal probabilities are computed by summing the relevant joint propensities for each village. For each village, the eight joint propensities sum to one.
\end{flushleft}
\end{table}

\subsection{Detailed results}

We now present the results of our proposed estimators. We organize the results into two parts: Table~\ref{tab:conditional_effects} reports the conditional direct and indirect effects, and Table~\ref{tab:marginal_effects} reports the in-policy marginal direct and indirect effects. For each estimand, we report three estimators: the Horvitz--Thompson estimator, the H\'{a}jek estimator, and the covariate-adjusted estimator. 

Overall, the H\'{a}jek estimator provides more precise estimates than the Horvitz--Thompson estimator, with standard errors typically 50--70\% smaller. This is consistent with the well-known finite sample efficiency advantages of H\'{a}jek estimation \citep{aronow2017estimating, ding2024first, gao2025causal}. Also, the results are relatively stable between the H\'{a}jek and the covariate-adjusted estimator in terms of both point and standard error estimators. Therefore, although we report all three estimators, our conclusions are mainly based on the H\'{a}jek and covariate-adjusted estimators, as we do not trust the Horvitz--Thompson estimator due to its instability.

\paragraph{Conditional causal effects.}
Panel A of Table~\ref{tab:conditional_effects} reports the estimated conditional direct effects $\hat{\de}^{*}(s,h)$, which measure the treatment effect of village $i$ receiving a cash transfer, conditional on its exposure environment $(S_i, H_i) = (s,h)$. We find significant positive effects on profits and revenues when villages are in high within-sublocation saturation but low geographic exposure environments $(S_i, H_i) = (1,0)$. In this setting, treated villages experience increases of $3,539$ KES in monthly profit ($p<0.01$) and $4,159$ KES in monthly revenue ($p<0.01$). This suggests that within-sublocation indirect effects create favorable conditions for treated enterprises to grow. By contrast, enterprises in villages with low local saturation but high geographic exposure $(S_i, H_i) = (0,1)$ face higher costs and wages, consistent with increased competition for inputs and labor. Effects are close to zero or negative when both exposure levels are simultaneously high or low, highlighting the importance of accounting for multiple indirect effect channels. Taken together, the patterns in profits, revenues, costs, and wages suggest two competing spillover forces. When exposure to treated villages in other sublocations is high, nearby villages may face stronger competition for workers and other inputs, which raises costs and can reduce net gains. When treatment saturation is high within the same sublocation but outside geographic exposure is low, villages may instead benefit from local economic spillovers, leading to higher profits and revenues. These are suggestive interpretations of the estimated effects, not direct evidence on the mechanism.

We next examine the conditional indirect effects, which isolate indirect effect channels by holding the village's own treatment status fixed. Panel B of Table~\ref{tab:conditional_effects} reports the within-sublocation indirect effects, which measure the impact of changing within-sublocation treatment saturation from low to high, conditional on the village's own treatment status $A_i=a$ and geographic exposure $H_i=h$. Among control villages, we find negative effects when geographic exposure is high $(H_i=1)$: control villages in high saturation sublocations with many nearby treated villages experience profit and revenue decreases, consistent with competitive pressure from treated neighbors. Among treated villages, we find large positive effects when geographic exposure is low $(H_i=0)$: treated villages benefit substantially from being in high saturation sublocations when they have few treated neighbors outside their sublocation. However, when geographic exposure is also high $(H_i=1)$, these benefits disappear.

Panel C of Table~\ref{tab:conditional_effects} presents the between-sublocation indirect effects, which measure the impact of increasing geographic exposure to treated villages in other sublocations from low to high, conditional on own treatment $A_i=a$ and within-sublocation saturation $S_i=s$. These between-sublocation indirect effects are heterogeneous across exposure configurations. Control villages in low saturation sublocations $(A_i,S_i)=(0,0)$ experience significant positive effects from greater geographic exposure to treated villages. This suggests that control villages benefit from proximity to treated villages, possibly through increased economic activity, though they also face higher input costs. However, among treated villages in high saturation sublocations $(A_i,S_i)=(1,1)$, the effects are strongly negative. This pattern is consistent with a tradeoff between positive local spillovers and negative competitive spillovers, with the latter becoming more salient when treated villages are surrounded by many treated villages in other sublocations.

\begin{table}[ht!]
\centering
\caption{Conditional direct and indirect effects}
\label{tab:conditional_effects}
\footnotesize
\setlength{\tabcolsep}{2.6pt} 
\renewcommand{\arraystretch}{1.05}

\begin{tabular}{@{}l*{4}{C{\colWw}}*{8}{C{\colW}}@{}}
\toprule
\multicolumn{13}{c}{\textbf{Panel A. Conditional direct effects:} $\hat{\de}^{*}(s,h)$} \\
\midrule
& \multicolumn{4}{c}{Horvitz--Thompson} & \multicolumn{4}{c}{H\'{a}jek} & \multicolumn{4}{c}{Covariate-adjusted} \\
\cmidrule(lr){2-5} \cmidrule(l){6-9} \cmidrule(l){10-13}
$(s,h)$ & Profit & Revenue & Costs & Wage & Profit & Revenue & Costs & Wage & Profit & Revenue & Costs & Wage \\
\midrule
$(0,0)$ & $2{,}142$ & $3{,}739$ & $542$ & $406$ & $-360$ & $-46$ & $200$ & $163$ & $-101$ & $217$ & $172$ & $136$ \\
        & $(2{,}180)$ & $(3{,}302)$ & $(352)$ & $(262)$ & $(909)$ & $(1{,}143)$ & $(135)$ & $(116)$ & $(883)$ & $(1{,}117)$ & $(137)$ & $(117)$ \\
\addlinespace
$(0,1)$ & $3{,}211$ & $5{,}841$ & $825^{*}$ & $625^{*}$ & $-558$ & $300$ & $329^{**}$ & $268^{*}$ & $-835$ & $-124$ & $260^{*}$ & $211$ \\
        & $(2{,}406)$ & $(3{,}940)$ & $(476)$ & $(361)$ & $(616)$ & $(877)$ & $(168)$ & $(145)$ & $(559)$ & $(696)$ & $(148)$ & $(131)$ \\
\addlinespace
$(1,0)$ & $272$ & $-440$ & $-300$ & $-253$ & $3{,}539^{***}$ & $4{,}159^{***}$ & $87$ & $23$ & $2{,}817^{**}$ & $3{,}043^{***}$ & $-150$ & $-181$ \\
        & $(3{,}431)$ & $(4{,}518)$ & $(469)$ & $(372)$ & $(1{,}263)$ & $(1{,}247)$ & $(169)$ & $(159)$ & $(1{,}268)$ & $(1{,}167)$ & $(157)$ & $(149)$ \\
\addlinespace
$(1,1)$ & $-3{,}233$ & $-4{,}785$ & $-458$ & $-340$ & $-227$ & $140$ & $115$ & $73$ & $37$ & $481$ & $142$ & $95$ \\
        & $(2{,}225)$ & $(3{,}364)$ & $(353)$ & $(260)$ & $(688)$ & $(825)$ & $(158)$ & $(134)$ & $(677)$ & $(776)$ & $(153)$ & $(131)$ \\
\addlinespace
\midrule
\multicolumn{13}{c}{\textbf{Panel B. Conditional within-sublocation indirect effects}: $\hat\wie^{*}(a,h)$} \\
\midrule
& \multicolumn{4}{c}{Horvitz--Thompson} & \multicolumn{4}{c}{H\'{a}jek} & \multicolumn{4}{c}{Covariate-adjusted} \\
\cmidrule(lr){2-5} \cmidrule(l){6-9} \cmidrule(l){10-13}
$(a,h)$ & Profit & Revenue & Costs & Wage & Profit & Revenue & Costs & Wage & Profit & Revenue & Costs & Wage \\
\midrule
$(0,0)$ & $2{,}881$ & $4{,}420$ & $584^{*}$ & $457^{*}$ & $-182$ & $-304$ & $132$ & $134$ & $67$ & $363$ & $333^{**}$ & $304^{***}$ \\
        & $(1{,}895)$ & $(2{,}813)$ & $(337)$ & $(262)$ & $(837)$ & $(961)$ & $(170)$ & $(148)$ & $(823)$ & $(943)$ & $(159)$ & $(140)$ \\
\addlinespace
$(0,1)$ & $3{,}078$ & $4{,}848$ & $541$ & $396$ & $-1{,}136^{*}$ & $-1{,}579^{*}$ & $-92$ & $-67$ & $-1{,}333^{**}$ & $-1{,}728^{**}$ & $-76$ & $-50$ \\
        & $(2{,}072)$ & $(3{,}131)$ & $(331)$ & $(248)$ & $(671)$ & $(813)$ & $(138)$ & $(121)$ & $(654)$ & $(750)$ & $(132)$ & $(116)$ \\
\addlinespace
$(1,0)$ & $1{,}011$ & $242$ & $-258$ & $-202$ & $3{,}718^{***}$ & $3{,}902^{***}$ & $18$ & $-6$ & $2{,}985^{**}$ & $3{,}189^{**}$ & $11$ & $-13$ \\
        & $(3{,}715)$ & $(5{,}007)$ & $(483)$ & $(372)$ & $(1{,}335)$ & $(1{,}429)$ & $(134)$ & $(127)$ & $(1{,}327)$ & $(1{,}342)$ & $(134)$ & $(126)$ \\
\addlinespace
$(1,1)$ & $-3{,}366$ & $-5{,}778$ & $-741$ & $-569$ & $-805$ & $-1{,}738^{**}$ & $-306$ & $-262^{*}$ & $-461$ & $-1{,}124$ & $-194$ & $-166$ \\
        & $(2{,}558)$ & $(4{,}173)$ & $(498)$ & $(373)$ & $(633)$ & $(889)$ & $(188)$ & $(158)$ & $(581)$ & $(722)$ & $(169)$ & $(147)$ \\
\addlinespace
\midrule
\multicolumn{13}{c}{\textbf{Panel C. Conditional between-sublocation indirect effects}: $\hat\bie^{*}(a,s)$} \\
\midrule
& \multicolumn{4}{c}{Horvitz--Thompson} & \multicolumn{4}{c}{H\'{a}jek} & \multicolumn{4}{c}{Covariate-adjusted} \\
\cmidrule(lr){2-5} \cmidrule(l){6-9} \cmidrule(l){10-13}
$(a,s)$ & Profit & Revenue & Costs & Wage & Profit & Revenue & Costs & Wage & Profit & Revenue & Costs & Wage \\
\midrule
$(0,0)$ & $-598$ & $-964$ & $-82$ & $-52$ & $1{,}330^{**}$ & $1{,}950^{***}$ & $222^{**}$ & $176^{**}$ & $1{,}442^{***}$ & $1{,}998^{***}$ & $204^{**}$ & $161^{**}$ \\
        & $(1{,}458)$ & $(2{,}190)$ & $(233)$ & $(177)$ & $(554)$ & $(645)$ & $(93)$ & $(83)$ & $(511)$ & $(555)$ & $(86)$ & $(76)$ \\
\addlinespace
$(0,1)$ & $-402$ & $-537$ & $-125$ & $-113$ & $375$ & $674$ & $-2$ & $-25$ & $42$ & $-92$ & $-204$ & $-194$ \\
        & $(2{,}509)$ & $(3{,}755)$ & $(435)$ & $(333)$ & $(954)$ & $(1{,}128)$ & $(214)$ & $(186)$ & $(967)$ & $(1{,}138)$ & $(205)$ & $(180)$ \\
\addlinespace
$(1,0)$ & $471$ & $1{,}138$ & $200$ & $167$ & $1{,}132$ & $2{,}296^{*}$ & $350^{*}$ & $282$ & $708$ & $1{,}658$ & $293^{*}$ & $236$ \\
        & $(3{,}127)$ & $(5{,}052)$ & $(594)$ & $(446)$ & $(971)$ & $(1{,}374)$ & $(209)$ & $(178)$ & $(930)$ & $(1{,}258)$ & $(199)$ & $(172)$ \\
\addlinespace
$(1,1)$ & $-3{,}906$ & $-4{,}882$ & $-283$ & $-201$ & $-3{,}391^{***}$ & $-3{,}344^{***}$ & $26$ & $25$ & $-2{,}739^{***}$ & $-2{,}654^{***}$ & $88$ & $83$ \\
        & $(3{,}146)$ & $(4{,}127)$ & $(387)$ & $(299)$ & $(997)$ & $(944)$ & $(112)$ & $(106)$ & $(978)$ & $(805)$ & $(104)$ & $(101)$ \\
\addlinespace
\bottomrule
\end{tabular}

\begin{minipage}{0.95\textwidth}
\vspace{2mm}
\footnotesize
\textit{Notes:} Each panel reports conditional causal effects estimated using Horvitz--Thompson, H\'{a}jek, and Covariate-adjusted estimators.  
Panel A reports conditional direct effects comparing treated ($A_i=1$) versus control ($A_i=0$) villages for each $(s,h)$.
Panel B reports conditional within-sublocation indirect effects comparing high ($S_i=1$) versus low ($S_i=0$) saturation levels for each $(a,h)$.  
Panel C reports between-sublocation indirect effects comparing high ($H_i=1$) versus low ($H_i=0$) exposure levels for each $(a,s)$.  
All monetary values are in Kenyan Shillings (KES) per enterprise per month.  
Point estimates are reported in the first line, with conservative standard error estimators in parentheses below.  
The covariate-adjusted estimator includes baseline covariates to improve efficiency. 
Significance levels: $^{*}p<0.10$, $^{**}p<0.05$, $^{***}p<0.01$.
\end{minipage}
\end{table}

\paragraph{In-policy marginal effects.} 
Table~\ref{tab:marginal_effects} reports estimates of the in-policy marginal direct and indirect effects, where the weighting scheme reflects the exposure distribution induced by the implemented policy. For simplicity, we only report the marginal indirect effects holding $A_i$ at control, i.e., $\wie(0)$ and $\bie(0)$.
In general, the in-policy marginal causal effects have relatively large estimated standard errors, and most estimates are not statistically significant. The lack of significance is consistent with the substantial heterogeneity documented in Table~\ref{tab:conditional_effects}: when the conditional causal effects vary strongly across exposure levels, marginalizing over these levels tends to dilute the signal and reduce statistical power. This also clarifies what is missed by analyses that do not explicitly model between-sublocation exposure: averaging over $H_i$ can attenuate or obscure the strong heterogeneity visible in the conditional effects.

\begin{table}[ht!]
\centering
\caption{In-policy marginal direct and indirect effects}
\label{tab:marginal_effects}
\footnotesize
\begin{tabular}{cccccccccccc}
\toprule
\multicolumn{4}{c}{Horvitz--Thompson} & \multicolumn{4}{c}{H\'{a}jek} & \multicolumn{4}{c}{Covariate-adjusted} \\
\cmidrule(lr){1-4} \cmidrule(lr){5-8} \cmidrule(lr){9-12}
Profit & Revenue & Costs & Wage & Profit & Revenue & Costs & Wage & Profit & Revenue & Costs & Wage \\
\midrule
\addlinespace
\multicolumn{12}{l}{\textbf{Marginal direct effect}: $\hat{\de}^{*}$} \\
\addlinespace
 $583$ & $1{,}191$ & $192$ & $140$ & $336$ & $811$ & $164$ & $121$ & $289$ & $656$ & $99$ & $64$ \\
 $(2{,}598)$ & $(3{,}930)$ & $(477)$ & $(370)$ & $(1{,}052)$ & $(1{,}242)$ & $(232)$ & $(205)$ & $(995)$ & $(1{,}119)$ & $(215)$ & $(191)$ \\
\addlinespace
\multicolumn{12}{l}{\textbf{Within-sublocation indirect effect}: $\hat{\wie}^{*}(0)$} \\
\addlinespace
 $2{,}886$ & $4{,}514$ & $612^{*}$ & $478^{*}$ & $-763$ & $-1{,}146$ & $53$ & $74$ & $-689$ & $-835$ & $142$ & $148$ \\
 $(1{,}914)$ & $(2{,}920)$ & $(359)$ & $(279)$ & $(710)$ & $(855)$ & $(168)$ & $(149)$ & $(686)$ & $(818)$ & $(155)$ & $(137)$ \\
\addlinespace
\multicolumn{12}{l}{\textbf{Between-sublocation indirect effect}: $\hat{\bie}^{*}(0)$} \\
\addlinespace
 $-296$ & $-362$ & $-18$ & $-8$ & $358$ & $601$ & $75$ & $61$ & $436$ & $565$ & $1$ & $-6$ \\
 $(2{,}850)$ & $(4{,}246)$ & $(462)$ & $(350)$ & $(1{,}138)$ & $(1{,}321)$ & $(217)$ & $(190)$ & $(1{,}060)$ & $(1{,}160)$ & $(206)$ & $(181)$ \\
\bottomrule
\end{tabular}
\begin{minipage}{0.95\textwidth}
\vspace{2mm}
\footnotesize
\textit{Notes:} This table reports in-policy marginal direct and indirect effects estimates. 
The marginal direct effect uses weights $\gamma_i(a,s,h) = \pr(S_i=s,H_i=h\mid A_i=a)$. 
The within-sublocation indirect effect compares $S_i=1$ to $S_i=0$ for control villages ($A_i=0$) using weights $\gamma_i(a,s,h) = \pr(H_i=h\mid A_i=a, S_i=s)$, while the between-sublocation indirect effect compares $H_i=1$ to $H_i=0$ for control villages ($A_i=0$) using weights $\gamma_i(a,s,h) = \pr(S_i=s\mid A_i=a, H_i=h)$. 
All monetary values are in Kenyan Shillings (KES) per enterprise per month. 
Point estimates are reported in the first line, with conservative standard error estimators in parentheses below. 
The covariate-adjusted estimator includes baseline covariates to improve efficiency. 
Significance levels: $^{*}p<0.10$, $^{**}p<0.05$, $^{***}p<0.01$.
\end{minipage}
\end{table}

\section{Discussion}
\label{sec::discussion}

Overall, our analysis has two features compared with the analysis in \cite{egger2022general}. First, we propose to decompose the spillover effect into within- and between-cluster spillover effects, which capture different sources of interference. There is no direct analogue of these two parts from the original paper. Second, we reveal heterogeneous effects in the empirical study \citet{egger2022general} conducted, while they mixed the two sources of interference components. Our results show that this decomposition yields new insights into the economic mechanisms underlying the cash transfer program.

There are several future directions to explore. First, beyond the nonparametric inverse propensity score weighting estimators developed in Section~\ref{sec::estimation}, it is also natural to consider regression-assisted approaches that partially pool information across exposure configurations. Because our estimands involve eight possible $(a,s,h)$ combinations, some cells may be sparsely populated in the real data. An unsaturated regression that restricts the three-way interaction and some two-way interactions among $(A_i,S_i,H_i)$ is of interest as a model-assisted perspective. The idea connects directly to the model-assisted framework of \citet{zhao2022reconciling} and to classical analyses of $2^3$ factorial designs. We leave it for future work. Second, throughout this paper, we have focused on binary exposure mappings for $(S_i, H_i)$, which leads to a finite number of exposure environments. This setting is conceptually simple and allows for nonparametric identification of all relevant average potential outcomes. The framework can be extended to categorical exposure mappings by defining additional exposure levels, for example, low/medium/high proportions of treated neighboring villages. However, the number of exposure cells grows quickly with the number of categories, and the corresponding number of average potential outcomes increases combinatorially. A more challenging extension involves allowing $(S_i, H_i)$ to take continuous values, such as the exact proportion of treated neighbors within and outside the sublocation. In this case, the object of interest becomes an exposure-response function $(a, s,h) \mapsto \bar Y(a,s,h)$, and recovering it requires either nonparametric methods or additional model structure. A fully nonparametric approach would require smoothing or kernel methods over a two-dimensional continuous exposure space, which may suffer from the curse of dimensionality and require a large sample size for stable estimation \citep{kennedy2017non}. Alternatively, one may impose parametric or semiparametric assumptions on the exposure–response relationship to obtain more stable inference at the cost of additional modeling assumptions. Third, in our real data analysis, we implemented an intuitive covariate‐adjustment strategy, analogous to Lin’s estimator \citep{lin2013agnostic} in the no‐interference setting, but without a formal justification for variance reduction under interference. Recent studies have begun to explore covariate adjustment under various settings with interference \citep{gao2025causal, lu2025design, chang2023design}. A natural direction for future research is to develop a rigorous model‐assisted covariate adjustment framework for interference settings, together with provable guarantees on efficiency gains. We leave it for future work. 

\bibliographystyle{apalike}
\bibliography{ref}

\begin{thebibliography}{}

\bibitem[Aronow and Samii, 2017]{aronow2017estimating}
Aronow, P.~M. and Samii, C. (2017).
\newblock Estimating average causal effects under general interference, with
  application to a social network experiment.
\newblock {\em The Annals of Applied Statistics}, 11(4):1912--1947.

\bibitem[Baird et~al., 2018]{baird2018optimal}
Baird, S., Bohren, J.~A., McIntosh, C., and {\"O}zler, B. (2018).
\newblock Optimal design of experiments in the presence of interference.
\newblock {\em Review of Economics and Statistics}, 100(5):844--860.

\bibitem[Bang and Robins, 2005]{bang2005doubly}
Bang, H. and Robins, J.~M. (2005).
\newblock Doubly robust estimation in missing data and causal inference models.
\newblock {\em Biometrics}, 61(4):962--973.

\bibitem[Basse and Feller, 2018]{basse2018analyzing}
Basse, G. and Feller, A. (2018).
\newblock Analyzing two-stage experiments in the presence of interference.
\newblock {\em Journal of the American Statistical Association},
  113(521):41--55.

\bibitem[Benjamin-Chung et~al., 2018]{benjamin2018spillover}
Benjamin-Chung, J., Arnold, B.~F., Berger, D., Luby, S.~P., Miguel, E.,
  Colford~Jr, J.~M., and Hubbard, A.~E. (2018).
\newblock Spillover effects in epidemiology: parameters, study designs and
  methodological considerations.
\newblock {\em International Journal of Epidemiology}, 47(1):332--347.

\bibitem[Chang, 2023]{chang2023design}
Chang, H. (2023).
\newblock Design-based estimation theory for complex experiments.
\newblock {\em arXiv preprint arXiv:2311.06891}.

\bibitem[Chen and Shao, 2004]{chen2004normal}
Chen, L.~H. and Shao, Q.-M. (2004).
\newblock {Normal approximation under local dependence}.
\newblock {\em The Annals of Probability}, 32(3A):1985--2028.

\bibitem[Crépon et~al., 2013]{crepon2013labor}
Crépon, B., Duflo, E., Gurgand, M., Rathelot, R., and Zamora, P. (2013).
\newblock Do labor market policies have displacement effects? evidence from a
  clustered randomized experiment.
\newblock {\em Quarterly Journal of Economics}, 128(2):531--580.

\bibitem[Ding, 2024]{ding2024first}
Ding, P. (2024).
\newblock {\em A First Course in Causal Inference}.
\newblock London: Chapman and Hall.

\bibitem[Egger et~al., 2022]{egger2022general}
Egger, D., Haushofer, J., Miguel, E., Niehaus, P., and Walker, M. (2022).
\newblock General equilibrium effects of cash transfers: experimental evidence
  from kenya.
\newblock {\em Econometrica}, 90(6):2603--2643.

\bibitem[Fogarty, 2018]{fogarty2018regression}
Fogarty, C.~B. (2018).
\newblock Regression-assisted inference for the average treatment effect in
  paired experiments.
\newblock {\em Biometrika}, 105(4):994--1000.

\bibitem[Gao and Ding, 2025]{gao2025causal}
Gao, M. and Ding, P. (2025).
\newblock Causal inference in network experiments: regression-based analysis
  and design-based properties.
\newblock {\em Journal of Econometrics}, 252:106119.

\bibitem[Giffin et~al., 2023]{giffin2023generalized}
Giffin, A., Reich, B., Yang, S., and Rappold, A. (2023).
\newblock Generalized propensity score approach to causal inference with
  spatial interference.
\newblock {\em Biometrics}, 79(3):2220--2231.

\bibitem[Horn and Johnson, 2012]{horn2012matrix}
Horn, R.~A. and Johnson, C.~R. (2012).
\newblock {\em Matrix analysis}.
\newblock Cambridge university press.

\bibitem[Hudgens and Halloran, 2008]{hudgens2008toward}
Hudgens, M.~G. and Halloran, M.~E. (2008).
\newblock Toward causal inference with interference.
\newblock {\em Journal of the American Statistical Association},
  103(482):832--842.

\bibitem[Jiang et~al., 2023]{jiang2023statistical}
Jiang, Z., Imai, K., and Malani, A. (2023).
\newblock Statistical inference and power analysis for direct and spillover
  effects in two-stage randomized experiments.
\newblock {\em Biometrics}, 79(3):2370--2381.

\bibitem[Kennedy et~al., 2017]{kennedy2017non}
Kennedy, E.~H., Ma, Z., McHugh, M.~D., and Small, D.~S. (2017).
\newblock Non-parametric methods for doubly robust estimation of continuous
  treatment effects.
\newblock {\em Journal of the Royal Statistical Society Series B: Statistical
  Methodology}, 79(4):1229--1245.

\bibitem[Leung, 2022]{leung2022causal}
Leung, M.~P. (2022).
\newblock Causal inference under approximate neighborhood interference.
\newblock {\em Econometrica}, 90(1):267--293.

\bibitem[Leung, 2025]{Leung10102025}
Leung, M.~P. (2025).
\newblock Cluster-randomized trials with cross-cluster interference.
\newblock {\em Journal of the American Statistical Association}, 0(ja):1--20.

\bibitem[Li and Wager, 2022]{li2022random}
Li, S. and Wager, S. (2022).
\newblock Random graph asymptotics for treatment effect estimation under
  network interference.
\newblock {\em The Annals of Statistics}, 50(4):2334--2358.

\bibitem[Lin, 2013]{lin2013agnostic}
Lin, W. (2013).
\newblock Agnostic notes on regression adjustments to experimental data:
  Reexamining freedman’s critique.
\newblock {\em The Annals of Applied Statistics}, 7(1):295--318.

\bibitem[Lu et~al., 2025]{lu2025design}
Lu, S., Shi, L., Fang, Y., Zhang, W., and Ding, P. (2025).
\newblock Design-based causal inference in bipartite experiments.
\newblock {\em arXiv preprint arXiv:2501.09844}.

\bibitem[Melis et~al., 2005]{melis2005design}
Melis, R.~J., van Eijken, M.~I., Borm, G.~F., Wensing, M., Adang, E., van~de
  Lisdonk, E.~H., van Achterberg, T., and Olde~Rikkert, M.~G. (2005).
\newblock The design of the dutch easycare study: a randomised controlled trial
  on the effectiveness of a problem-based community intervention model for
  frail elderly people [nct00105378].
\newblock {\em BMC Health Services Research}, 5(1):65.

\bibitem[Mukerjee et~al., 2018]{mukerjee2018using}
Mukerjee, R., Dasgupta, T., and Rubin, D.~B. (2018).
\newblock Using standard tools from finite population sampling to improve
  causal inference for complex experiments.
\newblock {\em Journal of the American Statistical Association},
  113(522):868--881.

\bibitem[Neyman, 1990]{neyman1923application}
Neyman, J. (1923/1990).
\newblock On the application of probability theory to agricultural experiments.
  {E}ssay on principles. {S}ection 9.
\newblock {\em Statistical Science}, 5(4):465--472.

\bibitem[Papadogeorgou et~al., 2022]{papadogeorgou2022causal}
Papadogeorgou, G., Imai, K., Lyall, J., and Li, F. (2022).
\newblock Causal inference with spatio-temporal data: Estimating the effects of
  airstrikes on insurgent violence in iraq.
\newblock {\em Journal of the Royal Statistical Society: Series B (Statistical
  Methodology)}, 84(5):1969--1999.

\bibitem[Rubin, 1974]{rubin1974estimating}
Rubin, D.~B. (1974).
\newblock Estimating causal effects of treatments in randomized and
  nonrandomized studies.
\newblock {\em Journal of Educational Psychology}, 66(5):688--701.

\bibitem[Scharfstein et~al., 1999]{scharfstein1999adjusting}
Scharfstein, D.~O., Rotnitzky, A., and Robins, J.~M. (1999).
\newblock Adjusting for nonignorable drop-out using semiparametric nonresponse
  models.
\newblock {\em Journal of the American Statistical Association},
  94(448):1096--1120.

\bibitem[Sinclair et~al., 2012]{sinclair2012detecting}
Sinclair, B., McConnell, M., and Green, D.~P. (2012).
\newblock Detecting spillover effects: Design and analysis of multilevel
  experiments.
\newblock {\em American Journal of Political Science}, 56(4):1055--1069.

\bibitem[Su and Ding, 2021]{su2021model}
Su, F. and Ding, P. (2021).
\newblock Model-assisted analyses of cluster-randomized experiments.
\newblock {\em Journal of the Royal Statistical Society Series B: Statistical
  Methodology}, 83(5):994--1015.

\bibitem[Wang et~al., 2025]{wang2025design}
Wang, Y., Samii, C., Chang, H., and Aronow, P. (2025).
\newblock Design-based inference for spatial experiments under unknown
  interference.
\newblock {\em The Annals of Applied Statistics}, 19(1):744--768.

\bibitem[Zhao and Ding, 2022a]{zhao2022reconciling}
Zhao, A. and Ding, P. (2022a).
\newblock Reconciling design-based and model-based causal inferences for
  split-plot experiments.
\newblock {\em The Annals of Statistics}, 50(2):1170--1192.

\bibitem[Zhao and Ding, 2022b]{zhao2022regression}
Zhao, A. and Ding, P. (2022b).
\newblock Regression-based causal inference with factorial experiments:
  estimands, model specifications and design-based properties.
\newblock {\em Biometrika}, 109(3):799--815.

\end{thebibliography}

\newpage 
\pagenumbering{arabic} 
\renewcommand*{\thepage}{S\arabic{page}}

\begin{center}
  \LARGE {\bf Supplementary Material}
\end{center}
\appendix 

\pagenumbering{arabic} 
\renewcommand*{\thepage}{S\arabic{page}}

\setcounter{equation}{0} 
\global\long\def\theequation{S\arabic{equation}}%
 \setcounter{assumption}{0} 
\global\long\def\theassumption{S\arabic{assumption}}%
 \setcounter{theorem}{0} 
\global\long\def\thetheorem{S\arabic{theorem}}%
 \setcounter{proposition}{0} 
\global\long\def\theproposition{S\arabic{proposition}}%
 \setcounter{definition}{0} 
\global\long\def\thedefinition{S\arabic{definition}}%
 \setcounter{example}{0} 
\global\long\def\theexample{S\arabic{example}}%
 \setcounter{figure}{0} 
\global\long\def\thefigure{S\arabic{figure}}%
 \setcounter{table}{0} 
\global\long\def\thetable{S\arabic{table}}%

The supplementary material is organized as follows. Section \ref{sec:proof} contains the proofs for the theoretical results in the main paper. Section \ref{sec::additional_empirical_results} reports a set of additional empirical results for the \cite{egger2022general} study.

\section{Proof}\label{sec:proof}
\subsection{Lemmas}
In this section, we introduce a set of lemmas. 

\subsubsection{Bounds on the eigenvalues} 

Lemma \ref{lem:eigenvalue-bound} provides bounds on the eigenvalues of $\bLambda(a,s,h)$ and $\bLambda(a,s,h;a',s',h')$. 
\begin{lemma}[Eigenvalue bounds]\label{lem:eigenvalue-bound} Under Assumptions \ref{ass:positivity},  \ref{ass:degree} and \ref{ass:dependence}, we have
\begin{eqnarray*}
    \|\bLambda(a,s,h)\|_{\textup{op}} \leq \Delta^{2m} \cdot \left(\frac{1}{\underline{c}_{\pi}} + \frac{1}{\underline{c}_{\pi}^2}\right) < \infty, \\
    \|\bLambda(a,s,h;a',s',h')\|_{\textup{op}} \leq \Delta^{2m} \cdot \left(\frac{1}{\underline{c}_{\pi}} + \frac{1}{\underline{c}_{\pi}^2}\right) < \infty.
\end{eqnarray*}
\end{lemma}
\begin{proof}
    We first show an upper bound for the eigenvalues of the matrices $\bLambda(a,s,h)$ and $\bLambda(a,s,h;a',s',h')$. Take $\bLambda(a,s,h)$ as an example. Recall that 
    \begin{eqnarray*}
        \bLambda_{ij}(a,s,h) = \indicator\{j=i\} \frac{1 - \pscore{i}(a, s, h)}{\pscore{i}(a, s, h)} + \indicator\{j\neq i\}\frac{\pscore{ij}(a, s, h; a, s, h) - \pscore{i}(a, s, h)\pscore{j}(a, s, h)}{\pscore{i}(a, s, h)\pscore{j}(a, s, h)}.
    \end{eqnarray*}
    By Assumption \ref{ass:positivity}, we have that
    \begin{eqnarray*}
        |\bLambda_{ij}(a,s,h)| \leq \frac{1 - \pscore{i}(a, s, h)}{\pscore{i}(a, s, h)} + \frac{|\pscore{ij}(a, s, h; a, s, h) - \pscore{i}(a, s, h)\pscore{j}(a, s, h)|}{\pscore{i}(a, s, h)\pscore{j}(a, s, h)}
        \leq \frac{1}{\underline{c}_{\pi}} + \frac{1}{\underline{c}_{\pi}^2}.
    \end{eqnarray*}
    Also, based on Assumption \ref{ass:degree} and \ref{ass:dependence}, each row of $\bLambda(a,s,h)$ has at most $\Delta^{2m}$ non-zero entries. Therefore, applying the Gershgorin circle theorem \citep{horn2012matrix}, we have that
    \begin{eqnarray*}
        \|\bLambda(a,s,h)\|_{\textup{op}} \leq \Delta^{2m} \cdot \left(\frac{1}{\underline{c}_{\pi}} + \frac{1}{\underline{c}_{\pi}^2}\right) < \infty.
    \end{eqnarray*}
    Similarly, we have that
    \begin{eqnarray*}
        \|\bLambda(a,s,h;a',s',h')\|_{\textup{op}} \leq \Delta^{2m} \cdot \left(\frac{1}{\underline{c}_{\pi}} + \frac{1}{\underline{c}_{\pi}^2}\right) < \infty.
    \end{eqnarray*}
\end{proof}

\subsubsection{Berry--Esseen bound under graph dependency.}   We introduce a result from \cite{chen2004normal}. For self-consistency, we include the following lemma for a central limit theorem (CLT) under graph dependency. Consider a set of random variables $\{V_i, i \in \mathcal{V}\}$ indexed by the vertices of a graph $\mathcal{G}_{\textup{Dep}} = (\mathcal{V}, \mathcal{E})$. $\mathcal{G}_{\textup{Dep}}$ is said to be a dependency graph if, for any pair of disjoint sets $\Gamma_1$ and $\Gamma_2$ in $\mathcal{V}$ such that no edge in $\mathcal{E}$ has one endpoint in $\Gamma_1$ and the other in $\Gamma_2$, the sets of random variables $\{V_i, i \in \Gamma_1\}$ and $\{V_i, i \in \Gamma_2\}$ are independent. We specially marked the meaning of the graph for measuring dependency with a subscript ``$\textup{Dep}$'' in the notation to highlight its difference with the social or spatial network. 
\begin{lemma}[A Berry--Esseen bound under graph dependency, Theorem 2.7 in \cite{chen2004normal}]\label{lem::clt_graph_dependency}
    Let $D$ denote the maximal degree of $\mathcal{G}_{\textup{Dep}}$, that is, the maximal number of edges incident to a single vertex. Let $\{V_i, i \in \mathcal{V}\}$ be random variables indexed by the vertices of a dependency graph $\mathcal{G}_{\textup{Dep}} = (\mathcal{V}, \mathcal{E})$. Put $W = \sumi V_i$. Assume that $E(W^2) = 1$, $E(V_i) = 0$ and $E(|V_i|^p) \leq \theta^p$ for $i \in \mathcal{V}$ and for some $2 < p \le 3$ and $\theta > 0$. Then for any $z\in\mathbb{R}$,
    \begin{eqnarray*}
        |\pr(W \leq z) - \Phi(z)| \leq 75 D^{5(p-1)}|V|\theta^p. 
    \end{eqnarray*}
\end{lemma}

\paragraph{Variance bounds for a double-sum statistic over a dependency graph} Below, we follow the definition of a dependency graph, and study the variance bounds for a double-sum statistic that will be useful to establish variance estimation results later. 
\begin{lemma}[Variance bounds for a double-sum statistic over a dependency graph]\label{lem::variance_bounds}
    Let $\cG_{\textup{Dep}} = (\mathcal{V}, \mathcal{E})$ be a dependency graph, with adjacency matrix $K_n$ (with diagonal values as $1$) and maximum degree $D$. Let $(U_i, V_i)_{i\in\mathcal{V}}$ be a sequence of random variable pairs following $\cG_{\textup{Dep}}$, with mean zero and bounded fourth moments $\max\{E(U_i^4), E(V_i^4)\} \le \bar{\sigma}^4$. Let $(c_i)_{i\in\mathcal{V}}$ be a sequence of constants, with absolute values upper bounded by $\bar{c}$. Then,
    \begin{eqnarray*}
        \var\left\{\sumij (K_n)_{ij} \cdot c_{ij} U_i V_j\right\} \leq n D^3 \bar{c}^2 \bar{\sigma}^4.
    \end{eqnarray*}
\end{lemma}
\begin{proof}
    We have that
    \begin{eqnarray*}
        \var\left\{\sumij (K_n)_{ij} \cdot c_{ij} U_i V_j\right\} &=& \sumij \sumrs (K_n)_{ij} (K_n)_{rs} \cdot c_{ij} c_{rs} \cov(U_iV_j, U_rV_s).
    \end{eqnarray*}
    Here, $\cov(U_iV_j, U_rV_s) = 0$ if $U_iV_j \indep U_rV_s$. Also, the summands are zero if $(K_n)_{ij} = 0$ or $(K_n)_{rs} = 0$, which implies $U_i \nonindep V_j$, or $U_r \nonindep V_s$. Therefore, the index tuple $(i,j,r,s)\in[n]^4$ for the nonzero terms in the summation must satisfy:
    \begin{eqnarray*}
        U_i \nonindep V_j, \quad U_r \nonindep V_s, \quad U_iV_j \nonindep U_rV_s.
    \end{eqnarray*}
    On the dependency graph, this means: $i$ is connected to $j$, $r$ is connected to $s$, and at least one of $i,j$ is connected to one of $r,s$. Without loss of generality, if $j$ connects to $r$, this means we can find a path on the dependency graph in the form of $i\to j\to r \to s$ (some index might be the same, say $i = j$, as the dependency graph includes a self-loop). Conversely, any path with length no greater than three can lead to a potential nonzero index tuple. For each node $i$, there are at most $D^3$ potential paths with length no greater than three. Therefore, the total number of potential nonzero index tuples is at most $n D^3$.

    Now each term in the summation is upper-bounded by 
    \begin{eqnarray*}
        |c_{ij} c_{rs} \cov(U_iV_j, U_rV_s)| \leq \bar{c}^2 \sqrt{\var(U_iV_j) \var(U_rV_s)} \leq \bar{c}^2 \sqrt{\{E(U_i^4)+E(V_j^4)\}\{E(U_r^4)+E(V_s^4)\}/4} \le \bar{c}^2 \bar{\sigma}^4.
    \end{eqnarray*}
    Therefore, we conclude that 
    \begin{eqnarray*}
        \var\left\{\sumij (K_n)_{ij} \cdot c_{ij} U_i V_j\right\} \leq n D^3 \bar{c}^2 \bar{\sigma}^4.
    \end{eqnarray*}
    
\end{proof}

\subsection{Proof of Theorem \ref{thm:asym_var}}
The theorem largely follows from existing results in complex randomized experiments, such as \cite{aronow2017estimating}, \cite{leung2022causal}, and \cite{mukerjee2018using}, by defining pseudo potential outcomes using the weighting matrix $\Gamma$ and the original/centered potential outcomes. For completeness, we showcase the derivation of the asymptotic variance of the Horvitz--Thompson estimator for the average of potential outcomes at level $(a,s,h)$, which is given by
\begin{eqnarray*}
    && \var\left\{n^{-1}\sumi\frac{ \bbI_{i}(a,s,h)\cdot \gamma_i(a,s,h)}{ \pscore{i}(a,s,h)}Y_{i}\right\} \\
    & = & n^{-2}\sumi\frac{ \var\{\bbI_{i}(a,s,h)\}\gamma_i(a,s,h)^2Y_{i}(a,s,h)^2}{ \pscore{i}(a,s,h)^2} \\
    &&+ 
    n^{-2}\sumjneqi\frac{ \cov\{\bbI_{i}(a,s,h), \bbI_{j}(a,s,h)\} \cdot \gamma_i(a,s,h)\gamma_j(a,s,h)Y_{i}(a,s,h)Y_{j}(a,s,h)}{\pscore{i}(a,s,h)\pscore{j}(a,s,h)}.
\end{eqnarray*}
Using the variance formula for a Bernoulli variable, we have
\begin{eqnarray}\label{eqn:var-bbI}
    \var\{\bbI_{i}(a,s,h)\} = \pi_i(a,s,h)\{1-\pi_i(a,s,h)\}. 
\end{eqnarray}
Meanwhile, we can compute the covariance term between units $i$ and $j$ based on the definition:
\begin{eqnarray}
    \cov\{\bbI_{i}(a,s,h), \bbI_{j}(a,s,h)\} &=& 
    E\{\bbI_{i}(a,s,h) \bbI_{j}(a,s,h)\} - 
    E\{\bbI_{i}(a,s,h)\} E\{\bbI_{j}(a,s,h)\}\notag\\ 
    &=& \pi_{ij}(a,s,h;a,s,h) - \pi_i(a,s,h)\pi_j(a,s,h). \label{eqn:cov-bbI}
\end{eqnarray}
Equations \eqref{eqn:var-bbI} and \eqref{eqn:cov-bbI} together lead to the definition of $\Lambda_{ij}(a,s,h)$ and also complete the variance computation for the Horvitz--Thompson estimator.  

We now derive the asymptotic variance of the H\'ajek estimator. We have
\[
\hat{Y}^{\haj}(a,s,h;\Gamma) - \bar{Y}(a,s,h;\Gamma)
=
\frac{ n^{-1}\sumi \frac{\bbI_{i}(a,s,h)\gamma_i(a,s,h)}{\pscore{i}(a,s,h)}Y_{i}^{\haj} }
     { n^{-1}\sumi \frac{\bbI_{i}(a,s,h)\gamma_i(a,s,h)}{\pscore{i}(a,s,h)}\Big/ n^{-1}\sumi \gamma_i(a,s,h) }.
\]
Denote the numerator and denominator as:
\begin{eqnarray*}
    A_n(a,s,h) &=& n^{-1}\sum_{i=1}^n 
    \frac{\bbI_i(a,s,h)\,\gamma_i(a,s,h)}{\pi_i(a,s,h)} Y_i^{\haj}(a,s,h), \\
    B_n(a,s,h) &=& n^{-1}\sum_{i=1}^n 
    \frac{\bbI_i(a,s,h)\,\gamma_i(a,s,h)}{\pi_i(a,s,h)} \Big/ n^{-1}\sumi \gamma_i(a,s,h).
\end{eqnarray*}

Accordingly, we have $E\{A_n(a,s,h)\} = 0$ and $E\{B_n(a,s,h)\}=1$. We then have 
\begin{eqnarray*}
    \hat{Y}^{\haj}(a,s,h;\Gamma) - \bar{Y}(a,s,h;\Gamma) - A_n(a,s,h)&=& \frac{A_n(a,s,h)\cdot \{1-B_n(a,s,h)\}}{B_n(a,s,h)}.
\end{eqnarray*}
Now treating $1$ as pseudo potential outcomes and applying the variance formula for the Horvitz--Thompson estimator, we can show that 
\begin{eqnarray}
    \var\{B_n(a,s,h)\} = \frac{n^{-1} \bOneA^{\T} \bGamma(a,s,h) \ \bLambda(a,s,h) \bGamma(a,s,h) \bOneA}{\{n^{-1}\sumi \gamma_i(a,s,h)\}^2}. 
\end{eqnarray}
Using Lemma \ref{lem:eigenvalue-bound}, the numerator of the above formula can be bounded as 
\begin{eqnarray*}
    n^{-2} \bOneA^{\T} \bGamma(a,s,h) \ \bLambda(a,s,h) \bGamma(a,s,h) \bOneA \le n^{-2} \|\bGamma(a,s,h)\|_{\textup{op}}^2 \|\bLambda(a,s,h)\|_{\textup{op}} \|\bOneA\|_2^2 = O(n^{-1}).
\end{eqnarray*}
Hence, $\var(B_n) = O(n^{-1})$ and using Chebyshev's inequality, $B_n - 1 = o_p(1)$. Therefore, we can show that 
\begin{eqnarray*}
   \hat{Y}^{\haj}(a,s,h;\Gamma) - \bar{Y}(a,s,h;\Gamma) - A_n(a,s,h) = o_p(A_n(a,s,h)) = o_p\left(\var\{A_n(a,s,h)\}^{1/2}\right). 
\end{eqnarray*}
Or equivalently, 
\begin{eqnarray}\label{eqn:equvalence}
   \frac{\hat{Y}^{\haj}(a,s,h;\Gamma) - \bar{Y}(a,s,h;\Gamma)}{\var\{A_n(a,s,h)\}^{1/2}} - \frac{A_n(a,s,h)}{\var\{A_n(a,s,h)\}^{1/2}} = o_p(1).
\end{eqnarray}
Therefore, we can see $\hat{Y}^{\haj}(a,s,h;\Gamma) - \bar{Y}$ is asymptotically equivalent to $A_n(a,s,h)$, thus has the same asymptotic variance form as $\var\{A_n(a,s,h)\}$ which takes the form stated in the theorem. 

\qed

\paragraph{Scalar forms of the variance and covariance} In Theorem \ref{thm:asym_var}, we give a compact form of asymptotic variance and covariance for the reweighting estimators. We also note the more explicit forms below that appeared during the proofs, which are more straightforward and interpretable. For a reweighting regime $\Gamma$ and $*\in\{\horthom,\haj\}$, for a fixed treatment and exposure mapping level $(a,s,h)$, the asymptotic variance of $\hat{Y}^{*}(a,s,h;\Gamma)$ is
\begin{eqnarray*}
    && \avar\{\hat{Y}^{*}(a,s,h;\Gamma)\} \\
    &=& n^{-2} \bY^{*}(a,s,h)^{\T} \bGamma(a,s,h) \bLambda(a,s,h) \bGamma(a,s,h) \bY^{*}(a,s,h)\\
    &=& n^{-2}\sumi \frac{1 - \pscore{i}(a, s, h)}{\pscore{i}(a, s, h)}\{\gamma_i(a,s,h)Y_i^{*}(a, s, h)\}^2 \\
    &&+ n^{-2}\sumi\sumjneqi\frac{\pscore{ij}(a, s, h; a, s, h) - \pscore{i}(a, s, h)\pscore{j}(a, s, h)}{\pscore{i}(a, s, h)\pscore{j}(a, s, h)}\gamma_i(a,s,h)\gamma_j(a,s,h)Y_i^{*}(a, s, h)Y_j^{*}(a, s, h).
\end{eqnarray*}

Similarly, for a pair $(a,s,h)$ and $(a^{\prime},s^{\prime},h^{\prime})$, the asymptotic covariance is
\begin{eqnarray*}
    && \acov\{\hat{Y}^{*}(a,s,h;\Gamma), \hat{Y}^{*}(a',s',h';\Gamma)\} \\ 
    &=& n^{-2} \bY^*(a,s,h)^{\T} \bGamma(a,s,h)\bLambda(a,s,h;\ashprime) \bGamma(a',s',h')\bY^*(\ashprime)\\
    &=& -n^{-2}\sumi  \gamma_i(a,s,h)\gamma_i(\ashprime) Y_i^*(a, s, h) Y_i^*(\ashprime) \\ 
    &&+ n^{-2}\sumi\sumjneqi \frac{\pscore{ij}(a, s, h; \ashprime) - \pscore{i}(a, s, h)\pscore{j}(\ashprime)}{\pscore{i}(a, s, h)\pscore{j}(\ashprime)} \gamma_i(a,s,h)\gamma_j(\ashprime)Y_i^*(a, s, h) Y_j^*(\ashprime).
\end{eqnarray*}
This scalar form expression gives a more straightforward demonstration of the impossibility of consistent variance estimation due to the existence of the cross-level potential outcome terms.

\subsection{Proofs of Theorem \ref{thm:asymptotic} and Corollary \ref{cor:ci}}

\subsubsection{Consistency}
\begin{proof}
    
    We start by showing the consistency of the Horvitz--Thompson estimator. Using Lemma \ref{lem:eigenvalue-bound}, under Assumptions \ref{ass:bounded}-\ref{ass:dependence}, we have
    \begin{eqnarray*}
         &&n^{-1} \bY(a,s,h)^{\T} \bGamma(a,s,h) \ \bLambda(a,s,h) \bGamma(a,s,h) \bY(a,s,h) \\
         &\leq & \Delta^2 \cdot \left(\frac{1}{\underline{c}_{\pi}} + \frac{1}{\underline{c}_{\pi}^2}\right) \cdot \|\bGamma(a,s,h)\|_{\infty}^2 \cdot n^{-1}\|\bY(a,s,h)\|_2^2\\
         &\leq & \Delta^2 \cdot \left(\frac{1}{\underline{c}_{\pi}} + \frac{1}{\underline{c}_{\pi}^2}\right) \cdot C_Y^2 < \infty.
    \end{eqnarray*}
    Similarly, we have that
    \begin{eqnarray*}
        &&n^{-1} \bY(a,s,h)^{\T} \bGamma(a,s,h) \ \bLambda(a,s,h;a',s',h') \bGamma(a',s',h') \bY(a',s',h') \\
        &\leq & \Delta^2 \cdot \left(\frac{1}{\underline{c}_{\pi}} + \frac{1}{\underline{c}_{\pi}^2}\right) \cdot \|\bGamma(a,s,h)\|_{\infty} \cdot \|\bGamma(a',s',h')\|_{\infty} \cdot n^{-1}\|\bY(a',s',h')\|_2^2\\
        &\leq & \Delta^2 \cdot \left(\frac{1}{\underline{c}_{\pi}} + \frac{1}{\underline{c}_{\pi}^2}\right) \cdot C_Y^2 < \infty.
    \end{eqnarray*}
    Therefore, based on Theorem \ref{thm:asym_var}, we have that
    \begin{eqnarray*}
        \var\{\hat{Y}^{\horthom}(a,s,h;\Gamma)\} = n^{-2} \bY(a,s,h)^{\T} \bGamma(a,s,h) \ \bLambda(a,s,h) \bGamma(a,s,h) \bY(a,s,h) = O(n^{-1}).
    \end{eqnarray*}
    Similarly, we have that
    \begin{eqnarray*}
        \cov\{\hat{Y}^{\horthom}(a,s,h;\Gamma), \hat{Y}^{\horthom}(a',s',h';\Gamma)\} = O(n^{-1}).
    \end{eqnarray*}
    Therefore, by Chebyshev's inequality, we have that $\hat{\bY}^{\horthom}_{\Gamma} - \bar{\bY}_{\Gamma} = O_p(n^{-1/2})$. Moreover, we can show the consistency of the H\'{a}jek estimator by further showing the denominator is consistent to $1$ following a similar argument, then applying the continuous mapping theorem. 
\end{proof}

\subsubsection{Asymptotic normality} 
Below, we use Lemma \ref{lem::clt_graph_dependency} to show CLTs for $\hat{\bY}^{\horthom}_{\Gamma}$ and $\hat{\bY}^{\haj}_{\Gamma}$ in Theorem \ref{thm:asymptotic}, Part 2. 

\begin{proof}
    \textbf{Step 1:} We first show the asymptotic normality of $\hat{\bY}^{\horthom}_{\Gamma}$. We will use Lemma \ref{lem::clt_graph_dependency} along with the Cramer-Wold device to show this. Consider a weighted combination of the standardized estimator given a weight vector $\boldsymbol{w}$ with $\|\boldsymbol{w}\|_2 = 1$:
    \begin{eqnarray*}
        W &=& \boldsymbol{w}^\T  \avar(\hat{\bY}^{\horthom}_{\Gamma})^{-1/2} (\hat{\bY}^{\horthom}_{\Gamma} - \bar{\bY}_{\Gamma}) \\
        &=& \boldsymbol{w}^\T_{\Gamma} (\hat{\bY}^{\horthom}_{\Gamma} - \bar{\bY}_{\Gamma})\quad (\text{by the definition of $\boldsymbol{w}_{\Gamma} = \avar(\hat{\bY}^{\horthom}_{\Gamma})^{-1/2}\boldsymbol{w}$})\\
        &=& n^{-1}\sumi \sum_{a,s,h}\boldsymbol{w}_{\Gamma}(a,s,h) \left\{\frac{\indicator_i(a,s,h)\cdot\gamma_i(a,s,h) Y_i(a,s,h)}{\pi_i(a,s,h)} - \gamma_i(a,s,h) Y_i(a,s,h)\right\}.
    \end{eqnarray*}
    Then $E(W) = 0$ and $\var(W) = 1$. Note that due to the non-degeneracy assumption in \eqref{eqn:non-degeneracy}, we have
    \begin{eqnarray*}
        \boldsymbol{w}_{\Gamma} = \avar(\hat{\bY}^{\horthom}_{\Gamma})^{-1/2} \boldsymbol{w} = \Theta(n^{1/2}).
    \end{eqnarray*}

    Taking $V_i$ in Lemma \ref{lem::clt_graph_dependency} as the scaled random component:
    \begin{eqnarray*}
        V_i = n^{-1} \sum_{a,s,h}\boldsymbol{w}_{\Gamma}(a,s,h) \left\{\frac{\indicator_i(a,s,h)\cdot\gamma_i(a,s,h) Y_i(a,s,h)}{\pi_i(a,s,h)} - \gamma_i(a,s,h) Y_i(a,s,h)\right\}, 
    \end{eqnarray*}
    with $E(V_i) = 0$. Now we bound $E(|V_i|^3)$. By Assumption \ref{ass:bounded}, we have that
    \begin{eqnarray*}
        E(|V_i|^3) \leq Cn^{-3}\sum_{a,s,h} \boldsymbol{w}_{\Gamma}(a,s,h)^3 \cdot |\gamma_i(a,s,h) Y_i(a,s,h)|^3 \cdot \left|\frac{\indicator_i(a,s,h)}{\pi_i(a,s,h)} - 1\right|^3 \le CC_Y^3 \cdot \underline{c}_{\pi}^{-3} n^{-3/2}.
    \end{eqnarray*}
    Based on the definition for dependency graph, we have can define $\cG_{\textup{Dep}} = \indicator\{\cB^{2m} > 0\}$ to be the adjacency matrix that indicates $2m$-order neighbors based on the original village network graph $\cB$. By Assumption \ref{ass:degree}, the maximal degree of $\cG_{\textup{Dep}}$ is at most $\Delta^{2m}$. Also $|V| = n$. Therefore, using Lemma \ref{lem::clt_graph_dependency}, we have that
    \begin{eqnarray*}
        |\pr(W \leq z) - \Phi(z)| \leq C C_Y^3 \underline{c}_{\pi}^{-3} \Delta^{20m} n^{-1/2}. 
    \end{eqnarray*}
     Using the Cram\'{e}r--Wold device, we have that
     \begin{eqnarray*}
        \avar(\hat{\bY}^{\horthom}_{\Gamma})^{-1/2} (\hat{\bY}^{\horthom}_{\Gamma} - \bar{\bY}_{\Gamma}) \rightarrow \cN(0, I_{8}).
     \end{eqnarray*}

     \textbf{Step 2:} Similar arguments can be applied to the H\'{a}jek estimator, by using the asymptotic equivalence result \eqref{eqn:equvalence},  replacing the potential outcomes used in Step 1 with the centered potential outcomes, and applying Slutsky's theorem. 
    
\end{proof}

\subsection{Conservative variance estimator and confidence intervals}

\subsubsection{A general asymptotic result on conservative variance estimator and confidence intervals}
In this section, we provide a theoretical argument for the conservative variance estimator and confidence intervals for the proposed estimators, from which we can derive Corollary \ref{cor:ci} directly. 

\begin{theorem}[Conservative variance estimator and confidence intervals]\label{thm:ci}
    Assume Assumptions \ref{ass:bounded}--\ref{ass:dependence}. Assume the nondegeneracy of asymptotic variance. Then the variance estimators under a single level $(a,s,h)$ for both Horvitz--Thompson and Hajek estimators are consistent:
    \begin{eqnarray*}
        \frac{\hat{\var}\{\hat{Y}^{*}(a,s,h;\Gamma)\}}{\avar\{\hat{Y}^{*}(a,s,h;\Gamma)\}} = 1 + O_p(n^{-1/2}), \text{ for } *\in\{\horthom, \haj\}.
    \end{eqnarray*}
    
    Moreover, for any linear combination of the single-level estimators under weight $w$, define a variance estimator
    \begin{eqnarray*}
        \hat{\var}(w^\T \hat{\bY}^*_{\Gamma}) = \sum_{a,s,h;a',s',h'} |w(a,s,h)| |w(a',s',h')|\cdot  \hat{\var}\{\hat{Y}^*(a,s,h)\}^{1/2} \cdot \hat{\var}\{\hat{Y}^*(a',s',h')\}^{1/2}. 
    \end{eqnarray*}
    Then we have that the variance estimator is conservative:
    \begin{eqnarray*}
        \frac{\hat{\var}(w^\T \hat{\bY}^*_{\Gamma})}{\bar{\var}(w^\T \hat{\bY}^*_{\Gamma})} = 1 + o_p(1),\quad \bar{\var}(w^\T \hat{\bY}^*_{\Gamma}) \ge \var(w^\T  \hat{\bY}^{*}_{\Gamma}).
    \end{eqnarray*}
    where
    \begin{eqnarray*}
        \bar{\var}(w^\T \hat{\bY}^*_{\Gamma}) = \sum_{a,s,h;a',s',h'} |w(a,s,h)| |w(a',s',h')| \cdot \avar\{\hat{Y}^*(a,s,h)\} \cdot \avar\{\hat{Y}^*(a',s',h')\}.
    \end{eqnarray*}
\end{theorem}
\begin{proof}
    \textbf{Step 1: variance estimation for the Horvitz--Thompson estimator for a single level of the exposure mapping.} Define $K_n = \indicator(B^{2m} > 0)$ as the $2m$-order adjacency matrix over graph $\cB$. We have that
    \begin{eqnarray*}
        && n^{-1} \hat{\bY}^{\horthom}(a,s,h)^\T \bGamma(a,s,h)\bOmega(a,s,h) \bGamma(a,s,h) \hat{\bY}^{\horthom}(a,s,h)\\
        &=& n^{-1} \sumij (K_{n})_{ij} \cdot \frac{\indicator_{ij}(a,s,h;a,s,h)}{\pi_{ij}(a,s,h;a,s,h)} \cdot \bLambda_{ij}(a,s,h) \gamma_i(a,s,h)\gamma_j(a,s,h)Y_i(a,s,h)Y_j(a,s,h). 
    \end{eqnarray*}
    We can see that 
    \begin{eqnarray*}
        E\left[ n \cdot \hat{\textup{se}}^2\{\hat{Y}^{\horthom}(a,s,h;\Gamma)\}\right] = n \cdot \avar\{\hat{Y}^{\horthom}(a,s,h;\Gamma)\} = \Theta(1). 
    \end{eqnarray*}
    Applying Lemma \ref{lem::variance_bounds}, we have
    \begin{eqnarray*}
        n \cdot \hat{\textup{se}}^2\{\hat{Y}^{\horthom}(a,s,h;\Gamma)\} - n \cdot \avar\{\hat{Y}^{\horthom}(a,s,h;\Gamma)\} = O_p(n^{-1/2}), \text{ where } n \cdot \avar\{\hat{Y}^{\horthom}(a,s,h;\Gamma)\} = \Theta(1). 
    \end{eqnarray*}

    Now using $n \cdot \avar\{\hat{Y}^{\horthom}(a,s,h;\Gamma)\} = \Theta(1)$, we have
    \begin{eqnarray*}
        \frac{\hat{\var}\{\hat{Y}^{\horthom}(a,s,h;\Gamma)\}}{\avar\{\hat{Y}^{\horthom}(a,s,h;\Gamma)\}} = 1 + O_p(n^{-1/2}). 
    \end{eqnarray*}   

    \textbf{Step 2: variance estimation for the H\'{a}jek estimator for a single level of the exposure mapping.} Variance estimation for the H\'{a}jek estimator can be justified mostly in the same way as the Horvitz--Thompson estimator. The only difference is that we need to handle the plug-in center. Define the aggregated vector $\hat{\bY}^{\haj, *} = (\hat{Y}_1^{\haj, *}, \ldots, \hat{Y}_n^{\haj, *})^{\T}$, where $\hat{Y}_i^{\haj, *} = \bbI_{i}(a,s,h) \tilde{Y}^{*}_i / \pscore{i}(a,s,h)$ and 
    $$\tilde{Y}^*_i = Y_i - \frac{\bar{Y}(a,s,h;\Gamma)}{n^{-1}\sumi \gamma_i(a,s,h)}.$$
    Then we have the following decomposition: 
    \begin{eqnarray*}
        && n^{-1} \hat{\bY}^{\haj}(a,s,h)^\T \bGamma(a,s,h)\bOmega(a,s,h) \bGamma(a,s,h) \hat{\bY}^{\haj}(a,s,h)\\
        &=& \underbrace{n^{-1} \hat{\bY}^{\haj, *}(a,s,h)^\T \bGamma(a,s,h)\bOmega(a,s,h) \bGamma(a,s,h) \hat{\bY}^{\haj, *}(a,s,h)}_{\text{Term I}}\\
        && 
        + \underbrace{\frac{n^{-1}}{\{n^{-1}\sumi \gamma_i(a,s,h)\}^2} \cdot \hat{\indicator}(a,s,h)^\T \bGamma(a,s,h) \bOmega(a,s,h)\bGamma(a,s,h) \hat{\indicator}(a,s,h) \left\{\hat{Y}^{\haj}(a,s,h) - \bar{Y}(a,s,h)\right\}^2}_{\text{Term II}} \\
        && - \underbrace{\frac{n^{-1}}{n^{-1}\sumi \gamma_i(a,s,h)} \cdot \hat{\bY}^{\haj, * }(a,s,h)^\T \bGamma(a,s,h) \bOmega(a,s,h)\bGamma(a,s,h) \hat{\indicator}(a,s,h) \left\{\hat{Y}^{\haj}(a,s,h) - \bar{Y}(a,s,h)\right\}}_{\text{Term III}}.  
    \end{eqnarray*}

    For Term I, similar to the Horvitz--Thompson estimator, we can show that 
    \begin{eqnarray*}
        \text{Term I} - n \cdot \avar\{\hat{Y}^{\haj}(a,s,h;\Gamma)\} = O_p(n^{-1/2}), \text{ where } n \cdot \avar\{\hat{Y}^{\haj}(a,s,h;\Gamma)\} = \Theta(1). 
    \end{eqnarray*}   
    For Term II and III, using Lemma \ref{lem::variance_bounds} again, we have that
    \begin{eqnarray*}
        \frac{n^{-1}}{\{n^{-1}\sumi \gamma_i(a,s,h)\}^2} \cdot \hat{\indicator}(a,s,h)^\T \bGamma(a,s,h) \bOmega(a,s,h)\bGamma(a,s,h) \hat{\indicator}(a,s,h) &=& O_p(1), \\
        \frac{n^{-1}}{n^{-1}\sumi \gamma_i(a,s,h)} \cdot \hat{\bY}^{\haj, * }(a,s,h)^\T \bGamma(a,s,h) \bOmega(a,s,h)\bGamma(a,s,h) \hat{\indicator}(a,s,h) &=& O_p(1). 
    \end{eqnarray*}
    By Theorem \ref{thm:asymptotic}, 
    \begin{eqnarray*}
        \hat{Y}^{\haj}(a,s,h;\Gamma) - \bar{Y}(a,s,h;\Gamma) = o_p(1).
    \end{eqnarray*}
    Therefore, 
    \begin{eqnarray*}
        n \cdot \hat{\var}\{\hat{Y}^{\haj}(a,s,h;\Gamma)\} &=& \text{Term I} + \text{Term II} + \text{Term III} \\
        &=& n \cdot \avar\{\hat{Y}^{\haj}(a,s,h;\Gamma)\} + o_p(1). 
    \end{eqnarray*}
    Equivalently, 
    \begin{eqnarray*}
        \frac{\hat{\var}\{\hat{Y}^{\haj}(a,s,h;\Gamma)\}}{\avar\{\hat{Y}^{\haj}(a,s,h;\Gamma)\}} = 1 + o_p(1). 
    \end{eqnarray*}

    \textbf{Step 3: variance estimation for any linear combination of the single-level estimators. } 
    Now, to estimate the variance for any linear combination of the single-level estimators, given weights say $w$, we have that 
    \begin{eqnarray*}
        \var(w^\T \hat{\bY}^*_{\Gamma}) &=& \sum_{a,s,h;a',s',h'} w(a,s,h) w(a',s',h') \cov\{\hat{Y}^*(a,s,h), \hat{Y}^*(a',s',h')\}\\
        &\le & \sum_{a,s,h;a',s',h'} |w(a,s,h)| |w(a',s',h')|\cdot  \var\{\hat{Y}^*(a,s,h)\}^{1/2} \cdot \var\{\hat{Y}^*(a',s',h')\}^{1/2} = \bar{\var}(w^\T \hat{\bY}^*_{\Gamma}),
    \end{eqnarray*}
    where we applied the Cauchy--Schwarz inequality. Therefore, we can estimate the above upper bound with:
    \begin{eqnarray*}
        \hat{\var}(w^\T \hat{\bY}^*_{\Gamma}) = \sum_{a,s,h;a',s',h'} |w(a,s,h)| |w(a',s',h')|\cdot  \hat{\var}\{\hat{Y}^*(a,s,h)\}^{1/2} \cdot \hat{\var}\{\hat{Y}^*(a',s',h')\}^{1/2}. 
    \end{eqnarray*}
    And by Step 1 and Step 2, we have the convergence result:
    \begin{eqnarray*}
        \frac{\hat{\var}(w^\T \hat{\bY}^*_{\Gamma})}{\bar{\var}(w^\T \hat{\bY}^*_{\Gamma})} = 1 + o_p(1). 
    \end{eqnarray*}

\end{proof}

\subsubsection{Proof of Corollary \ref{cor:ci}} 
We now use Theorem \ref{thm:ci} to prove Corollary \ref{cor:ci}. 
\begin{enumerate}
    \item  For Part 1, using Theorem \ref{thm:ci}, 
\begin{eqnarray*}
    \frac{n\cdot [\hat{\textup{se}}^2\{\hat{Y}^{*}(a,s,h;\Gamma)\} - \avar\{\hat{Y}^{*}(a,s,h;\Gamma)\}]}{n\cdot \avar\{\hat{Y}^{*}(a,s,h;\Gamma)\}} = O_p(n^{-1/2}), \text{ for } *\in\{\horthom, \haj\}.
\end{eqnarray*}
Also using our condition that $n\cdot \avar\{\hat{Y}^{*}(a,s,h;\Gamma)\} = \Theta(1)$, we conclude the convergence of the variance estimator. Now applying Slutsky's Theorem with the asymptotic normality results in Theorem \ref{thm:asymptotic}, we can prove the asymptotic validity of the confidence interval. 

   \item For Part 2, we can follow a similar proof logic as Part 1. The only difference is that we would apply the conservative variance estimation conclusion from Theorem \ref{thm:ci} with the asymptotic normality results in Theorem \ref{thm:asymptotic}, which also establishes valid asymptotic coverage. 
\end{enumerate} 


\section{Additional empirical results}
\label{sec::additional_empirical_results}

In this section, we report additional empirical results. Subsection \ref{sec::robustness} reports several numerical studies to check the robustness of our conclusion to the choice of geographical closeness in defining exposure mappings measured by distance and number of nearby villages. Subsection \ref{sec::cutoff} examines the sensitivity of our results to the choice of the cutoff treated proportion used to dichotomize the exposure
variables. Subsection \ref{sec::no_H} investigates the role of between-sublocation interference and highlights the importance of considering it besides the within-sublocation interference. 

\subsection{Robustness to distance and number of nearby villages}\label{sec::robustness}

In this section, we examine the robustness of our empirical results to the choice of geographical closeness used in constructing the between-sublocation exposure variable $H_i$. Our baseline specification uses a distance threshold of 4 km (strict) with $k=3$ nearest villages outside the sublocation. We consider two additional specifications: (i) distance $\leq 6$ km with $k=5$ nearest villages, and (ii) distance $\leq 6$ km with $k=8$ nearest villages. In all cases, we use the threshold of $1/2$ for defining $S_i$ and $H_i$. These alternative specifications should be interpreted as different coarse summaries of the underlying interference process. If the true mechanism is not perfectly captured by a specific threshold, each specification defines a valid estimand under that exposure mapping, and robustness is assessed by whether the substantive conclusions are stable across plausible mappings.

Table~\ref{tab:conditional_robustness} reports the H\'{a}jek estimates of conditional direct and indirect effects across the three specifications. Several patterns emerge. First, the conditional direct effect at $(S_i, H_i)=(1,0)$ is large, positive, and highly significant for profits and revenues across all three specifications, with point estimates ranging from 2{,}971 to 3{,}706 KES for profits and 3{,}478 to 4{,}401 KES for revenues. This finding is robust to the choice of nearby village definition, providing strong evidence that treated villages in high within-sublocation saturation environments with low geographic exposure benefit substantially from the cash transfer.

Second, the between-sublocation indirect effects in Panel C show consistent patterns. Among treated villages in high saturation sublocations $(A_i, S_i)=(1,1)$, the effect of increasing geographic exposure is strongly negative across all three specifications, with profit effects ranging from $-2{,}731$ to $-3{,}391$ KES, all significant at the 1\% level. Among control villages in low saturation sublocations $(A_i, S_i)=(0,0)$, the between-sublocation effects on costs and wages are positive and significant across all three specifications. The revenue effect is also positive and significant under two of the three specifications, while the profit effect is significant only under the baseline specification.

Table~\ref{tab:marginal_robustness} reports the in-policy marginal effects. Consistent with the main analysis, none of the marginal effects are statistically significant under any specification, reflecting the heterogeneity in conditional effects that is masked by marginalization.

\begin{table}[ht!]
\centering
\caption{Robustness of conditional direct and indirect effects (H\'{a}jek estimator)}
\label{tab:conditional_robustness}
\footnotesize
\setlength{\tabcolsep}{2.8pt}
\renewcommand{\arraystretch}{1.05}

\begin{tabular}{@{}l*{4}{C{\colWw}}*{8}{C{\colW}}@{}}
\toprule
\multicolumn{13}{c}{\textbf{Panel A. Conditional direct effects:} $\hat{\de}^{\haj}(s,h)$} \\
\midrule
& \multicolumn{4}{c}{dist$<$4, $k$=3} & \multicolumn{4}{c}{dist$\leq$6, $k$=5} & \multicolumn{4}{c}{dist$\leq$6, $k$=8} \\
\cmidrule(lr){2-5} \cmidrule(l){6-9} \cmidrule(l){10-13}
$(s,h)$ & Profit & Revenue & Costs & Wage & Profit & Revenue & Costs & Wage & Profit & Revenue & Costs & Wage \\
\midrule
$(0,0)$ & $-360$ & $-46$ & $200$ & $163$ & $92$ & $504$ & $361^{**}$ & $309^{**}$ & $-173$ & $622$ & $406^{***}$ & $344^{***}$ \\
        & $(909)$ & $(1{,}143)$ & $(135)$ & $(116)$ & $(925)$ & $(1{,}148)$ & $(160)$ & $(140)$ & $(712)$ & $(863)$ & $(127)$ & $(111)$ \\
\addlinespace
$(0,1)$ & $-558$ & $300$ & $329^{**}$ & $268^{*}$ & $201$ & $738$ & $174$ & $130$ & $-384$ & $-102$ & $230$ & $196$ \\
        & $(616)$ & $(877)$ & $(168)$ & $(145)$ & $(527)$ & $(1{,}045)$ & $(246)$ & $(203)$ & $(883)$ & $(1{,}475)$ & $(348)$ & $(297)$ \\
\addlinespace
$(1,0)$ & $3{,}539^{***}$ & $4{,}159^{***}$ & $87$ & $23$ & $3{,}706^{***}$ & $4{,}401^{***}$ & $88$ & $18$ & $2{,}971^{***}$ & $3{,}478^{***}$ & $-7$ & $-53$ \\
        & $(1{,}263)$ & $(1{,}247)$ & $(169)$ & $(159)$ & $(717)$ & $(797)$ & $(179)$ & $(156)$ & $(694)$ & $(844)$ & $(173)$ & $(141)$ \\
\addlinespace
$(1,1)$ & $-227$ & $140$ & $115$ & $73$ & $112$ & $89$ & $-169$ & $-175$ & $-5$ & $374$ & $-28$ & $-81$ \\
        & $(688)$ & $(825)$ & $(158)$ & $(134)$ & $(954)$ & $(1{,}154)$ & $(240)$ & $(208)$ & $(1{,}279)$ & $(1{,}283)$ & $(363)$ & $(342)$ \\
\addlinespace
\midrule
\multicolumn{13}{c}{\textbf{Panel B. Conditional within-sublocation indirect effects}: $\hat\wie^{\haj}(a,h)$} \\
\midrule
& \multicolumn{4}{c}{dist$<$4, $k$=3} & \multicolumn{4}{c}{dist$\leq$6, $k$=5} & \multicolumn{4}{c}{dist$\leq$6, $k$=8} \\
\cmidrule(lr){2-5} \cmidrule(l){6-9} \cmidrule(l){10-13}
$(a,h)$ & Profit & Revenue & Costs & Wage & Profit & Revenue & Costs & Wage & Profit & Revenue & Costs & Wage \\
\midrule
$(0,0)$ & $-182$ & $-304$ & $132$ & $134$ & $-626$ & $-1{,}081$ & $117$ & $135$ & $-366$ & $-190$ & $249$ & $230$ \\
        & $(837)$ & $(961)$ & $(170)$ & $(148)$ & $(558)$ & $(751)$ & $(180)$ & $(156)$ & $(507)$ & $(675)$ & $(164)$ & $(143)$ \\
\addlinespace
$(0,1)$ & $-1{,}136^{*}$ & $-1{,}579^{*}$ & $-92$ & $-67$ & $-503$ & $-793$ & $35$ & $51$ & $-969$ & $-1{,}900$ & $-80$ & $-24$ \\
        & $(671)$ & $(813)$ & $(138)$ & $(121)$ & $(953)$ & $(1{,}139)$ & $(219)$ & $(191)$ & $(1{,}376)$ & $(1{,}283)$ & $(351)$ & $(337)$ \\
\addlinespace
$(1,0)$ & $3{,}718^{***}$ & $3{,}902^{***}$ & $18$ & $-6$ & $2{,}988^{***}$ & $2{,}816^{**}$ & $-156$ & $-155$ & $2{,}779^{***}$ & $2{,}667^{***}$ & $-164$ & $-167$ \\
        & $(1{,}335)$ & $(1{,}429)$ & $(134)$ & $(127)$ & $(1{,}084)$ & $(1{,}194)$ & $(159)$ & $(140)$ & $(899)$ & $(1{,}032)$ & $(136)$ & $(108)$ \\
\addlinespace
$(1,1)$ & $-805$ & $-1{,}738^{*}$ & $-306$ & $-262^{*}$ & $-593$ & $-1{,}443$ & $-308$ & $-254$ & $-591$ & $-1{,}424$ & $-338$ & $-301$ \\
        & $(633)$ & $(889)$ & $(188)$ & $(158)$ & $(528)$ & $(1{,}059)$ & $(267)$ & $(220)$ & $(786)$ & $(1{,}476)$ & $(359)$ & $(302)$ \\
\addlinespace
\midrule
\multicolumn{13}{c}{\textbf{Panel C. Conditional between-sublocation indirect effects}: $\hat\bie^{\haj}(a,s)$} \\
\midrule
& \multicolumn{4}{c}{dist$<$4, $k$=3} & \multicolumn{4}{c}{dist$\leq$6, $k$=5} & \multicolumn{4}{c}{dist$\leq$6, $k$=8} \\
\cmidrule(lr){2-5} \cmidrule(l){6-9} \cmidrule(l){10-13}
$(a,s)$ & Profit & Revenue & Costs & Wage & Profit & Revenue & Costs & Wage & Profit & Revenue & Costs & Wage \\
\midrule
$(0,0)$ & $1{,}330^{**}$ & $1{,}950^{***}$ & $221^{**}$ & $176^{**}$ & $486$ & $900$ & $285^{***}$ & $241^{***}$ & $849$ & $2{,}024^{***}$ & $438^{***}$ & $372^{***}$ \\
        & $(554)$ & $(645)$ & $(93)$ & $(83)$ & $(525)$ & $(677)$ & $(97)$ & $(81)$ & $(564)$ & $(591)$ & $(90)$ & $(80)$ \\
\addlinespace
$(0,1)$ & $375$ & $674$ & $-2$ & $-25$ & $609$ & $1{,}187$ & $202$ & $157$ & $245$ & $314$ & $109$ & $118$ \\
        & $(954)$ & $(1{,}128)$ & $(214)$ & $(186)$ & $(986)$ & $(1{,}214)$ & $(302)$ & $(266)$ & $(1{,}320)$ & $(1{,}367)$ & $(426)$ & $(400)$ \\
\addlinespace
$(1,0)$ & $1{,}132$ & $2{,}296^{*}$ & $350^{*}$ & $282$ & $595$ & $1{,}134$ & $98$ & $62$ & $639$ & $1{,}301$ & $262$ & $225$ \\
        & $(971)$ & $(1{,}374)$ & $(209)$ & $(178)$ & $(926)$ & $(1{,}516)$ & $(309)$ & $(261)$ & $(1{,}031)$ & $(1{,}747)$ & $(385)$ & $(327)$ \\
\addlinespace
$(1,1)$ & $-3{,}391^{***}$ & $-3{,}344^{***}$ & $26$ & $25$ & $-2{,}986^{***}$ & $-3{,}125^{***}$ & $-54$ & $-37$ & $-2{,}731^{***}$ & $-2{,}790^{***}$ & $88$ & $91$ \\
        & $(997)$ & $(944)$ & $(112)$ & $(106)$ & $(685)$ & $(737)$ & $(117)$ & $(99)$ & $(654)$ & $(760)$ & $(111)$ & $(83)$ \\
\addlinespace
\bottomrule
\end{tabular}

\begin{minipage}{0.95\textwidth}
\vspace{2mm}
\footnotesize
\textit{Notes:} Each panel reports conditional causal effects estimated using the H\'{a}jek estimator under three specifications of geographical closeness: (i) distance $<$ 4 km with $k=3$ nearest villages, (ii) distance $\leq$ 6 km with $k=5$ nearest villages, and (iii) distance $\leq$ 6 km with $k=8$ nearest villages.
Panel A reports conditional direct effects comparing treated ($A_i=1$) versus control ($A_i=0$) villages for each $(s,h)$.
Panel B reports conditional within-sublocation indirect effects comparing high ($S_i=1$) versus low ($S_i=0$) saturation levels for each $(a,h)$.  
Panel C reports between-sublocation indirect effects comparing high ($H_i=1$) versus low ($H_i=0$) exposure levels for each $(a,s)$.  
All monetary values are in Kenyan Shillings (KES) per enterprise per month.  
Point estimates are reported in the first line, with conservative standard error estimators in parentheses below.  
Significance levels: $^{*}p<0.10$, $^{**}p<0.05$, $^{***}p<0.01$.
\end{minipage}
\end{table}

\begin{table}[ht!]
\centering
\caption{Robustness of in-policy marginal direct and indirect effects (H\'{a}jek estimator)}
\label{tab:marginal_robustness}
\footnotesize
\begin{tabular}{cccccccccccc}
\toprule
\multicolumn{4}{c}{dist$<$4, $k$=3} & \multicolumn{4}{c}{dist$\leq$6, $k$=5} & \multicolumn{4}{c}{dist$\leq$6, $k$=8} \\
\cmidrule(lr){1-4} \cmidrule(lr){5-8} \cmidrule(lr){9-12}
Profit & Revenue & Costs & Wage & Profit & Revenue & Costs & Wage & Profit & Revenue & Costs & Wage \\
\midrule
\addlinespace
\multicolumn{12}{l}{\textbf{Marginal direct effect}: $\hat{\de}^{*}$} \\
\addlinespace
 $336$ & $811$ & $164$ & $121$ & $467$ & $962$ & $161$ & $117$ & $383$ & $879$ & $171$ & $126$ \\
 $(1{,}052)$ & $(1{,}242)$ & $(232)$ & $(205)$ & $(969)$ & $(1{,}207)$ & $(244)$ & $(213)$ & $(955)$ & $(1{,}198)$ & $(233)$ & $(201)$ \\
\addlinespace
\multicolumn{12}{l}{\textbf{Within-sublocation indirect effect}: $\hat{\wie}^{*}(0)$} \\
\addlinespace
 $-763$ & $-1{,}146$ & $53$ & $74$ & $-761$ & $-1{,}151$ & $45$ & $68$ & $-607$ & $-944$ & $59$ & $77$ \\
 $(710)$ & $(855)$ & $(168)$ & $(149)$ & $(664)$ & $(831)$ & $(180)$ & $(159)$ & $(645)$ & $(795)$ & $(181)$ & $(160)$ \\
\addlinespace
\multicolumn{12}{l}{\textbf{Between-sublocation indirect effect}: $\hat{\bie}^{*}(0)$} \\
\addlinespace
 $358$ & $601$ & $75$ & $61$ & $691$ & $1{,}111$ & $227$ & $190$ & $919$ & $1{,}464$ & $201$ & $181$ \\
 $(1{,}138)$ & $(1{,}321)$ & $(217)$ & $(190)$ & $(978)$ & $(1{,}259)$ & $(247)$ & $(210)$ & $(1{,}153)$ & $(1{,}299)$ & $(259)$ & $(233)$ \\
\addlinespace
\bottomrule
\end{tabular}
\begin{minipage}{0.95\textwidth}
\vspace{2mm}
\footnotesize
\textit{Notes:} This table reports in-policy marginal direct and indirect effects estimated using the H\'{a}jek estimator under three specifications of geographic closeness: (i) distance $<$ 4 km with $k=3$ nearest villages, (ii) distance $\leq$ 6 km with $k=5$ nearest villages, and (iii) distance $\leq$ 6 km with $k=8$ nearest villages. 
All monetary values are in Kenyan Shillings (KES) per enterprise per month. 
Point estimates are reported in the first line, with conservative standard error estimators in parentheses below. 
Significance levels: $^{*}p<0.10$, $^{**}p<0.05$, $^{***}p<0.01$.

\end{minipage}
\end{table}

\subsection{Different cutoffs for defining the exposure mappings}\label{sec::cutoff}

We next examine the sensitivity of our results to the choice of the cutoff value $c$ used to dichotomize the exposure variables. Recall that the baseline specification uses $c = 1/2$, so that $S_i = \indicator\{\text{fraction treated in sublocation} > 1/2\}$ and $H_i = \indicator\{\text{fraction treated among neighbors not in the same sublocation} > 1/2\}$. We consider two alternative cutoffs: $c = 1/3$ (a lower threshold, so that $(S_i,H_i)=(1,1)$ are more common) and $c = 2/3$ (a higher threshold, so that $(S_i,H_i)=(1,1)$ are rarer). The network construction criterion is held fixed at a distance $< 4$ km with $k = 3$ nearest villages throughout.

Table~\ref{tab:conditional_cutoff} reports the conditional effects across the three cutoff values. In Panel A, the conditional direct effect at $(S_i,H_i)=(1,0)$ is positive and significant for profits and revenues under both $c = 1/3$ (2{,}013 and 2{,}634 KES) and $c = 1/2$ ($3{,}539$ and $4{,}159$ KES), but becomes small and insignificant at $c = 2/3$. This pattern is intuitive: with a higher cutoff, $S_i = 1$ requires a very large fraction of treated villages in the sublocation, and the resulting exposure cell becomes sparsely populated, reducing power. At the same time, $\de(1,1)$ becomes significantly negative at $c = 2/3$ ($-1{,}587$ for profits, $-1{,}977$ for revenues), suggesting that when both within-sublocation and geographic exposure are very high, treated enterprises may face competitive pressure.

In Panel B, the within-sublocation indirect effects show consistent patterns across cutoff values. At $(A_i, H_i)=(1,0)$, the effect of increased sublocation saturation is positive for profits and revenues at $c = 1/3$ and $c = 1/2$, but reverses at $c = 2/3$. Notably, the effects at $(A_i, H_i)=(1,1)$ become strongly negative and significant at $c = 2/3$, with profits declining by 1{,}935 KES and revenues by 4{,}798 KES.

In Panel C, the between-sublocation indirect effects at $(A_i,S_i)=(1,1)$ are consistently negative across all cutoffs, with the strongest effects at the baseline $c = 1/2$ ($-3{,}391$ for profits). The positive effects on costs at $(A_i,S_i)=(0,0)$ are present under all three cutoffs.

Table~\ref{tab:marginal_cutoff} reports the marginal effects. None of the marginal effects are statistically significant at any cutoff, consistent with the main analysis.

\begin{table}[ht!]
\centering
\caption{Robustness of conditional direct and indirect effects to cutoff choice (H\'{a}jek estimator)}
\label{tab:conditional_cutoff}
\footnotesize
\setlength{\tabcolsep}{2.8pt}
\renewcommand{\arraystretch}{1.05}

\begin{tabular}{@{}l*{4}{C{\colWw}}*{8}{C{\colW}}@{}}
\toprule
\multicolumn{13}{c}{\textbf{Panel A. Conditional direct effects:} $\hat{\de}^{\haj}(s,h)$} \\
\midrule
& \multicolumn{4}{c}{$c = 1/3$} & \multicolumn{4}{c}{$c = 1/2$} & \multicolumn{4}{c}{$c = 2/3$} \\
\cmidrule(lr){2-5} \cmidrule(l){6-9} \cmidrule(l){10-13}
$(s,h)$ & Profit & Revenue & Costs & Wage & Profit & Revenue & Costs & Wage & Profit & Revenue & Costs & Wage \\
\midrule
$(0,0)$ & $-1{,}086$ & $-1{,}028$ & $-15$ & $-39$ & $-360$ & $-46$ & $200$ & $163$ & $511$ & $983$ & $170$ & $129$ \\
        & $(746)$ & $(944)$ & $(144)$ & $(128)$ & $(909)$ & $(1{,}143)$ & $(135)$ & $(116)$ & $(646)$ & $(756)$ & $(146)$ & $(132)$ \\
\addlinespace
$(0,1)$ & $572$ & $1{,}159$ & $275$ & $215$ & $-558$ & $300$ & $329^{**}$ & $268^{*}$ & $43$ & $1{,}774$ & $469$ & $277$ \\
        & $(601)$ & $(887)$ & $(181)$ & $(157)$ & $(616)$ & $(877)$ & $(168)$ & $(145)$ & $(1{,}080)$ & $(1{,}824)$ & $(593)$ & $(482)$ \\
\addlinespace
$(1,0)$ & $2{,}013^{*}$ & $2{,}634^{**}$ & $238$ & $183$ & $3{,}539^{***}$ & $4{,}159^{***}$ & $87$ & $23$ & $-582$ & $-575$ & $-107$ & $-144$ \\
        & $(1{,}136)$ & $(1{,}228)$ & $(194)$ & $(177)$ & $(1{,}263)$ & $(1{,}247)$ & $(169)$ & $(159)$ & $(547)$ & $(721)$ & $(198)$ & $(178)$ \\
\addlinespace
$(1,1)$ & $-798$ & $-81$ & $233$ & $186$ & $-227$ & $140$ & $115$ & $73$ & $-1{,}587^{**}$ & $-1{,}977^{**}$ & $-265$ & $-242$ \\
        & $(773)$ & $(972)$ & $(155)$ & $(138)$ & $(688)$ & $(825)$ & $(158)$ & $(134)$ & $(807)$ & $(865)$ & $(177)$ & $(164)$ \\
\addlinespace
\midrule
\multicolumn{13}{c}{\textbf{Panel B. Conditional within-sublocation indirect effects}: $\hat\wie^{\haj}(a,h)$} \\
\midrule
& \multicolumn{4}{c}{$c = 1/3$} & \multicolumn{4}{c}{$c = 1/2$} & \multicolumn{4}{c}{$c = 2/3$} \\
\cmidrule(lr){2-5} \cmidrule(l){6-9} \cmidrule(l){10-13}
$(a,h)$ & Profit & Revenue & Costs & Wage & Profit & Revenue & Costs & Wage & Profit & Revenue & Costs & Wage \\
\midrule
$(0,0)$ & $-613$ & $-711$ & $-65$ & $-67$ & $-182$ & $-304$ & $132$ & $134$ & $-563$ & $-462$ & $204$ & $202$ \\
        & $(839)$ & $(970)$ & $(142)$ & $(125)$ & $(837)$ & $(961)$ & $(170)$ & $(148)$ & $(632)$ & $(829)$ & $(246)$ & $(223)$ \\
\addlinespace
$(0,1)$ & $439$ & $92$ & $-63$ & $-60$ & $-1{,}136^{*}$ & $-1{,}579^{*}$ & $-92$ & $-67$ & $-305$ & $-1{,}047$ & $-58$ & $-63$ \\
        & $(709)$ & $(805)$ & $(135)$ & $(122)$ & $(671)$ & $(813)$ & $(138)$ & $(121)$ & $(1{,}179)$ & $(1{,}297)$ & $(363)$ & $(333)$ \\
\addlinespace
$(1,0)$ & $2{,}486^{**}$ & $2{,}951^{**}$ & $188$ & $156$ & $3{,}718^{***}$ & $3{,}902^{***}$ & $18$ & $-6$ & $-1{,}655^{***}$ & $-2{,}020^{***}$ & $-73$ & $-71$ \\
        & $(1{,}043)$ & $(1{,}202)$ & $(196)$ & $(179)$ & $(1{,}335)$ & $(1{,}429)$ & $(134)$ & $(127)$ & $(561)$ & $(648)$ & $(98)$ & $(87)$ \\
\addlinespace
$(1,1)$ & $-931$ & $-1{,}148$ & $-104$ & $-89$ & $-805$ & $-1{,}738^{*}$ & $-306$ & $-262^{*}$ & $-1{,}935^{***}$ & $-4{,}798^{***}$ & $-792^{*}$ & $-582^{*}$ \\
        & $(665)$ & $(1{,}053)$ & $(201)$ & $(174)$ & $(633)$ & $(889)$ & $(188)$ & $(158)$ & $(709)$ & $(1{,}393)$ & $(408)$ & $(313)$ \\
\addlinespace
\midrule
\multicolumn{13}{c}{\textbf{Panel C. Conditional between-sublocation indirect effects}: $\hat\bie^{\haj}(a,s)$} \\
\midrule
& \multicolumn{4}{c}{$c = 1/3$} & \multicolumn{4}{c}{$c = 1/2$} & \multicolumn{4}{c}{$c = 2/3$} \\
\cmidrule(lr){2-5} \cmidrule(l){6-9} \cmidrule(l){10-13}
$(a,s)$ & Profit & Revenue & Costs & Wage & Profit & Revenue & Costs & Wage & Profit & Revenue & Costs & Wage \\
\midrule
$(0,0)$ & $-29$ & $643$ & $166^{**}$ & $126^{*}$ & $1{,}330^{**}$ & $1{,}950^{***}$ & $221^{**}$ & $176^{**}$ & $-543$ & $181$ & $395$ & $392^{*}$ \\
        & $(389)$ & $(499)$ & $(80)$ & $(68)$ & $(554)$ & $(645)$ & $(93)$ & $(83)$ & $(567)$ & $(677)$ & $(246)$ & $(223)$ \\
\addlinespace
$(0,1)$ & $1{,}023$ & $1{,}446$ & $168$ & $134$ & $375$ & $674$ & $-2$ & $-25$ & $-286$ & $-403$ & $134$ & $128$ \\
        & $(1{,}159)$ & $(1{,}276)$ & $(197)$ & $(179)$ & $(954)$ & $(1{,}128)$ & $(214)$ & $(186)$ & $(1{,}243)$ & $(1{,}449)$ & $(363)$ & $(332)$ \\
\addlinespace
$(1,0)$ & $1{,}629^{*}$ & $2{,}830^{**}$ & $456^{*}$ & $381^{*}$ & $1{,}132$ & $2{,}296^{*}$ & $350^{*}$ & $282$ & $-1{,}011$ & $973$ & $695$ & $541$ \\
        & $(958)$ & $(1{,}331)$ & $(245)$ & $(217)$ & $(971)$ & $(1{,}374)$ & $(209)$ & $(178)$ & $(1{,}159)$ & $(1{,}903)$ & $(494)$ & $(391)$ \\
\addlinespace
$(1,1)$ & $-1{,}788^{**}$ & $-1{,}268$ & $164$ & $136$ & $-3{,}391^{***}$ & $-3{,}344^{***}$ & $26$ & $25$ & $-1{,}291^{***}$ & $-1{,}805^{***}$ & $-24^{*}$ & $30^{***}$ \\
        & $(750)$ & $(924)$ & $(152)$ & $(136)$ & $(997)$ & $(944)$ & $(112)$ & $(106)$ & $(111)$ & $(137)$ & $(12)$ & $(9)$ \\
\addlinespace
\bottomrule
\end{tabular}

\begin{minipage}{0.95\textwidth}
\vspace{2mm}
\footnotesize
\textit{Notes:} Each panel reports conditional causal effects estimated using the H\'{a}jek estimator under three cutoff values $c \in \{1/3, 1/2, 2/3\}$ for defining the binary exposure variables $S_i = \indicator\{\text{fraction treated in sublocation} > c\}$ and $H_i = \indicator\{\text{fraction treated among neighbors not in the same sublocation} > c\}$. The network construction criterion is held fixed at a distance $<$ 4 km with $k=3$ nearest villages.
Panel A reports conditional direct effects comparing treated ($A_i=1$) versus control ($A_i=0$) villages for each $(s,h)$.
Panel B reports conditional within-sublocation indirect effects comparing high ($S_i=1$) versus low ($S_i=0$) saturation levels for each $(a,h)$.  
Panel C reports between-sublocation indirect effects comparing high ($H_i=1$) versus low ($H_i=0$) exposure levels for each $(a,s)$.  
All monetary values are in Kenyan Shillings (KES) per enterprise per month.  
Point estimates are reported in the first line, with conservative standard error estimators in parentheses below.  
Significance levels: $^{*}p<0.10$, $^{**}p<0.05$, $^{***}p<0.01$.
\end{minipage}
\end{table}

\begin{table}[ht!]
\centering
\caption{Robustness of in-policy marginal effects to cutoff choice (H\'{a}jek estimator)}
\label{tab:marginal_cutoff}
\footnotesize
\begin{tabular}{cccccccccccc}
\toprule
\multicolumn{4}{c}{$c = 1/3$} & \multicolumn{4}{c}{$c = 1/2$} & \multicolumn{4}{c}{$c = 2/3$} \\
\cmidrule(lr){1-4} \cmidrule(lr){5-8} \cmidrule(lr){9-12}
Profit & Revenue & Costs & Wage & Profit & Revenue & Costs & Wage & Profit & Revenue & Costs & Wage \\
\midrule
\addlinespace
\multicolumn{12}{l}{\textbf{Marginal direct effect}: $\hat{\de}^{*}$} \\
\addlinespace
 $495$ & $996$ & $168$ & $119$ & $336$ & $811$ & $164$ & $121$ & $61$ & $483$ & $164$ & $119$ \\
 $(1{,}018)$ & $(1{,}191)$ & $(229)$ & $(203)$ & $(1{,}052)$ & $(1{,}242)$ & $(232)$ & $(205)$ & $(870)$ & $(1{,}050)$ & $(199)$ & $(174)$ \\
\addlinespace
\multicolumn{12}{l}{\textbf{Within-sublocation indirect effect}: $\hat{\wie}^{*}(0)$} \\
\addlinespace
 $-158$ & $-358$ & $42$ & $32$ & $-763$ & $-1{,}146$ & $53$ & $74$ & $-511$ & $-476$ & $213$ & $206$ \\
 $(792)$ & $(882)$ & $(141)$ & $(126)$ & $(710)$ & $(855)$ & $(168)$ & $(149)$ & $(845)$ & $(1{,}043)$ & $(282)$ & $(256)$ \\
\addlinespace
\multicolumn{12}{l}{\textbf{Between-sublocation indirect effect}: $\hat{\bie}^{*}(0)$} \\
\addlinespace
 $92$ & $353$ & $82$ & $67$ & $358$ & $601$ & $75$ & $61$ & $-500$ & $-479$ & $191$ & $215$ \\
 $(1{,}030)$ & $(1{,}196)$ & $(197)$ & $(172)$ & $(1{,}138)$ & $(1{,}321)$ & $(217)$ & $(190)$ & $(1{,}089)$ & $(1{,}243)$ & $(270)$ & $(243)$ \\
\addlinespace
\bottomrule
\end{tabular}
\begin{minipage}{0.95\textwidth}
\vspace{2mm}
\footnotesize
\textit{Notes:} This table reports in-policy marginal direct and indirect effects estimated using the H\'{a}jek estimator under three cutoff values $c \in \{1/3, 1/2, 2/3\}$ for defining the binary exposure variables $S_i = \indicator\{\text{fraction treated in sublocation} > c\}$ and $H_i = \indicator\{\text{fraction treated among neighbors not in the same sublocation} > c\}$. The network construction criterion is held fixed at a distance $<$ 4 km with $k=3$ nearest villages.
All monetary values are in Kenyan Shillings (KES) per enterprise per month. 
Point estimates are reported in the first line, with conservative standard error estimators in parentheses below. 
Significance levels: $^{*}p<0.10$, $^{**}p<0.05$, $^{***}p<0.01$.
\end{minipage}
\end{table}

\subsection{Analysis ignoring between-sublocation interference}\label{sec::no_H}

To assess the role of between-sublocation interference in shaping our findings, we repeat the analysis using a reduced exposure mapping $(A_i, S_i)$ that ignores the geographic spillover variable $H_i$. Under this specification, each village is characterized by two binary indicators: its own treatment status $A_i$ and the within-sublocation spillover $S_i$, yielding four exposure cells instead of eight. The propensity scores $\pscore{i}(a,s)$ are estimated via the same Monte Carlo simulation of the experimental design, and the Horvitz--Thompson, H\'{a}jek, and covariate-adjusted estimators are constructed exactly as before but without the $H_i$ dimension. This exercise serves two purposes: it provides a benchmark that reflects what would be estimated under the standard two-stage randomized saturation design framework that ignores geographic proximity, and it helps isolate the extent to which accounting for between-sublocation interference alters inference about direct and within-sublocation indirect effects.

\paragraph{Conditional effects.}
Table~\ref{tab:conditional_effects_noH} reports conditional direct and within-sublocation indirect effects from the reduced $(A_i, S_i)$ model. 
The conditional direct effects exhibit a pattern qualitatively consistent with the full $(A_i,S_i,H_i)$ analysis: treatment effects on profits and revenue are not significant when the within-sublocation spillover is low ($S_i=0$), but are large and highly significant when $S_i=1$. Using the H\'{a}jek estimator, the direct effect on profits conditional on $S_i=1$ is $1{,}925$ KES ($p < 0.01$), and on revenue is $2{,}427$ KES ($p < 0.01$). The conditional indirect effects also show strong heterogeneity: among treated villages ($A_i=1$), the indirect effect of moving from $S_i=0$ to $S_i=1$ is large and positive for profits ($1{,}648$ KES, $p < 0.01$) and revenue ($1{,}421$ KES, $p < 0.05$), while among control villages ($A_i=0$) the indirect effects are negative but imprecisely estimated.

\paragraph{Comparison with the full $(A_i,S_i,H_i)$ model.}
While the qualitative direction of effects is preserved, there are notable differences relative to Table~\ref{tab:conditional_effects}. In the full model, the conditional direct effects are reported separately for each $(s,h)$ pair, revealing considerable heterogeneity: for example, the H\'{a}jek estimate of the direct effect is $3{,}539$ KES when $(S_i,H_i) = (1,0)$ but $-227$ KES when $(S_i,H_i) = (1,1)$. In the reduced model, these are averaged together into a single estimate for $S_i=1$, yielding $1{,}925$ KES, an intermediate value that masks the opposing effects of $H_i$. This averaging is analogous to the ``marginalizing over $H_i$'' operation, except that the reduced model cannot separately identify the $H$-specific components. The fact that the reduced model estimates are closer to the full-model marginal effects (Table~\ref{tab:marginal_effects}) than to any single conditional effect confirms that ignoring between-sublocation interference implicitly marginalizes over unmodeled geographic spillovers.

A second important difference concerns the between-sublocation indirect effects, which are large and significant in the full model (Table~\ref{tab:conditional_effects}, Panel~C) but are entirely absent from the reduced specification. In particular, the full model reveals that control villages in low saturation sublocations experience significant positive spillovers from geographic proximity to treated villages ($H$-effect of $1{,}330$ KES on profits, $p < 0.05$), while treated villages in high saturation sublocations experience large negative between-sublocation spillovers ($-3{,}391$ KES, $p < 0.01$). These effects, which are economically and statistically meaningful, are invisible under the reduced $(A_i,S_i)$ model and are instead absorbed into the error term, potentially biasing the within-sublocation estimates through omitted interference.

\begin{table}[ht!]
\centering
\caption{Conditional direct and indirect effects under the $(A_i,S_i)$ exposure mapping}
\label{tab:conditional_effects_noH}
\footnotesize
\setlength{\tabcolsep}{2.6pt} 
\renewcommand{\arraystretch}{1.05}

\begin{tabular}{@{}l*{4}{C{\colWw}}*{8}{C{\colW}}@{}}
\toprule
\multicolumn{13}{c}{\textbf{Panel A. Conditional direct effects:} $\hat\mu^{*}(1,s) - \hat\mu^{*}(0,s)$} \\
\midrule
& \multicolumn{4}{c}{Horvitz--Thompson} & \multicolumn{4}{c}{H\'{a}jek} & \multicolumn{4}{c}{Covariate-adjusted} \\
\cmidrule(lr){2-5} \cmidrule(l){6-9} \cmidrule(l){10-13}
$s$ & Profit & Revenue & Costs & Wage & Profit & Revenue & Costs & Wage & Profit & Revenue & Costs & Wage \\
\midrule
$0$ & $2{,}492^{*}$ & $4{,}382^{**}$ & $643^{***}$ & $493^{***}$ & $-298$ & $147$ & $252^{**}$ & $214^{**}$ & $-245$ & $230$ & $258^{***}$ & $219^{***}$ \\
    & $(1{,}341)$ & $(2{,}093)$ & $(239)$ & $(180)$ & $(428)$ & $(569)$ & $(104)$ & $(90)$ & $(432)$ & $(538)$ & $(94)$ & $(82)$ \\
\addlinespace
$1$ & $-1{,}465$ & $-2{,}620$ & $-435$ & $-351$ & $1{,}925^{***}$ & $2{,}427^{***}$ & $51$ & $-2$ & $1{,}534^{**}$ & $2{,}006^{***}$ & $12$ & $-37$ \\
    & $(2{,}079)$ & $(2{,}899)$ & $(312)$ & $(243)$ & $(641)$ & $(673)$ & $(132)$ & $(115)$ & $(650)$ & $(627)$ & $(131)$ & $(113)$ \\
\addlinespace
\midrule
\multicolumn{13}{c}{\textbf{Panel B. Conditional within-sublocation indirect effects}: $\hat\mu^{*}(a,1) - \hat\mu^{*}(a,0)$} \\
\midrule
& \multicolumn{4}{c}{Horvitz--Thompson} & \multicolumn{4}{c}{H\'{a}jek} & \multicolumn{4}{c}{Covariate-adjusted} \\
\cmidrule(lr){2-5} \cmidrule(l){6-9} \cmidrule(l){10-13}
$a$ & Profit & Revenue & Costs & Wage & Profit & Revenue & Costs & Wage & Profit & Revenue & Costs & Wage \\
\midrule
$0$ & $2{,}886^{**}$ & $4{,}514^{**}$ & $612^{***}$ & $478^{***}$ & $-575$ & $-859$ & $82$ & $96$ & $-321$ & $-425$ & $143$ & $144$ \\
    & $(1{,}258)$ & $(1{,}889)$ & $(231)$ & $(181)$ & $(436)$ & $(545)$ & $(122)$ & $(108)$ & $(438)$ & $(537)$ & $(119)$ & $(105)$ \\
\addlinespace
$1$ & $-1{,}071$ & $-2{,}489$ & $-467$ & $-366$ & $1{,}648^{***}$ & $1{,}421^{**}$ & $-119$ & $-120$ & $1{,}458^{**}$ & $1{,}351^{**}$ & $-103$ & $-112$ \\
    & $(2{,}162)$ & $(3{,}104)$ & $(320)$ & $(242)$ & $(632)$ & $(697)$ & $(114)$ & $(97)$ & $(644)$ & $(629)$ & $(106)$ & $(89)$ \\
\addlinespace
\bottomrule
\end{tabular}

\begin{minipage}{0.95\textwidth}
\vspace{2mm}
\footnotesize
\textit{Notes:} 
Each panel reports conditional causal effects estimated using Horvitz--Thompson, H\'{a}jek, and Covariate-adjusted estimators under the reduced $(A_i,S_i)$ exposure mapping that ignores between-sublocation interference.
Panel~A reports conditional direct effects comparing treated ($A_i=1$) versus control ($A_i=0$) villages for each $s$.
Panel~B reports conditional within-sublocation indirect effects comparing high ($S_i=1$) versus low ($S_i=0$) saturation levels for each $a$.
All monetary values are in Kenyan Shillings (KES) per enterprise per month.  
Point estimates are reported in the first line, with conservative standard error estimators in parentheses below.  
The covariate-adjusted estimator includes baseline covariates to improve efficiency. 
Significance levels: $^{*}p<0.10$, $^{**}p<0.05$, $^{***}p<0.01$.
\end{minipage}
\end{table}

\paragraph{Marginal effects.}
Table~\ref{tab:marginal_effects_noH} reports the in-policy marginal effects under the reduced model. The marginal direct effect on profits is $442$ KES (H\'{a}jek), comparable to the $336$ KES obtained in the full $(A_i,S_i,H_i)$ model. For the within-sublocation indirect effects, among treated villages ($A_i=1$), the reduced model yields a large positive effect on profits ($1{,}648$ KES, $p < 0.01$) and revenue ($1{,}421$ KES, $p < 0.05$), while among control villages ($A_i=0$) the indirect effects are negative but not statistically significant. These magnitudes broadly align with those from the full model, confirming that the within-sublocation spillover patterns are robust to the inclusion or exclusion of $H_i$.

The key implication of this comparison is that while the reduced $(A_i,S_i)$ model adequately captures the broad patterns of direct and within-sublocation effects, it misses an economically important dimension of interference. The full $(A_i,S_i,H_i)$ analysis reveals substantial heterogeneity in both direct and indirect effects along the $H_i$ dimension, and identifies significant between-sublocation spillovers that are entirely invisible under the standard framework. This underscores the value of incorporating geographic proximity into the exposure mapping.

\begin{table}[ht!]
\centering
\caption{In-policy marginal effects under the $(A_i,S_i)$ exposure mapping}
\label{tab:marginal_effects_noH}
\footnotesize
\begin{tabular}{cccccccccccc}
\toprule
\multicolumn{4}{c}{Horvitz--Thompson} & \multicolumn{4}{c}{H\'{a}jek} & \multicolumn{4}{c}{Covariate-adjusted} \\
\cmidrule(lr){1-4} \cmidrule(lr){5-8} \cmidrule(lr){9-12}
Profit & Revenue & Costs & Wage & Profit & Revenue & Costs & Wage & Profit & Revenue & Costs & Wage \\
\midrule
\addlinespace
\multicolumn{12}{l}{\textit{Marginal direct effect}} \\
\addlinespace
 $583$ & $1{,}191$ & $192$ & $140$ & $442$ & $948$ & $171$ & $126$ & $368$ & $813$ & $145$ & $105$ \\
 $(1{,}636)$ & $(2{,}471)$ & $(300)$ & $(233)$ & $(642)$ & $(779)$ & $(155)$ & $(135)$ & $(607)$ & $(707)$ & $(147)$ & $(129)$ \\
\addlinespace
\multicolumn{12}{l}{\textit{Within-sublocation indirect effect (control, $A_i=0$)}} \\
\addlinespace
 $2{,}886^{**}$ & $4{,}514^{**}$ & $612^{***}$ & $478^{***}$ & $-575$ & $-859$ & $82$ & $96$ & $-321$ & $-425$ & $143$ & $144$ \\
 $(1{,}258)$ & $(1{,}889)$ & $(231)$ & $(181)$ & $(436)$ & $(545)$ & $(122)$ & $(108)$ & $(438)$ & $(537)$ & $(119)$ & $(105)$ \\
\addlinespace
\multicolumn{12}{l}{\textit{Within-sublocation indirect effect (treated, $A_i=1$)}} \\
\addlinespace
 $-1{,}071$ & $-2{,}489$ & $-467$ & $-366$ & $1{,}648^{***}$ & $1{,}421^{**}$ & $-119$ & $-120$ & $1{,}458^{**}$ & $1{,}351^{**}$ & $-103$ & $-112$ \\
 $(2{,}162)$ & $(3{,}104)$ & $(320)$ & $(242)$ & $(632)$ & $(697)$ & $(114)$ & $(97)$ & $(644)$ & $(629)$ & $(106)$ & $(89)$ \\
\bottomrule
\end{tabular}
\begin{minipage}{0.95\textwidth}
\vspace{2mm}
\footnotesize
\textit{Notes:} This table reports in-policy marginal direct and indirect effects estimates under the reduced $(A_i,S_i)$ exposure mapping that ignores between-sublocation interference. 
The marginal direct effect uses weights $\gamma_i(a,s) = \pr(S_i=s\mid A_i=a)$. 
The within-sublocation indirect effect compares $S_i=1$ to $S_i=0$ for control, separately for control ($A_i=0$) and treated ($A_i=1$) villages, using uniform weights. 
All monetary values are in Kenyan Shillings (KES) per enterprise per month. 
Point estimates are reported in the first line, with conservative standard error estimators in parentheses below. 
The covariate-adjusted estimator includes baseline covariates to improve efficiency. 
Significance levels: $^{*}p<0.10$, $^{**}p<0.05$, $^{***}p<0.01$.
\end{minipage}
\end{table}

\end{document}